\documentclass[11pt]{article} 

\usepackage[caption=false]{subfig}
\usepackage[T1]{fontenc}
\usepackage[UKenglish]{babel}
\usepackage[UKenglish]{isodate}
\usepackage[utf8]{inputenc}
\usepackage[overload]{empheq}  
\usepackage[subnum]{cases}  
\usepackage{comment}
\usepackage{graphicx}
\usepackage[space]{grffile} % for spaces in include graphics
\usepackage{epstopdf}
\usepackage{pdfpages}
\usepackage{float} % to write H in the figures
\usepackage{tikz-cd} % diagrams
\usepackage{xcolor}
\RequirePackage[linkcolor = black,citecolor=blue,urlcolor=blue]{hyperref}
\usepackage{amsmath}
\usepackage{cleveref}
\usepackage{amssymb}
\usepackage{color}
\usepackage{amsthm}
\usepackage{booktabs}
\usepackage{bbm}
\usepackage{bm}
\usepackage{caption}
\usepackage{eucal}
\usepackage{float}
\usepackage{latexsym}
\usepackage{mathtools}
\usepackage{multicol}
\usepackage{multirow}
\usepackage{pdfpages}
\usepackage{siunitx}
\usepackage{authblk}

\usepackage{framed}

\usepackage{enumitem}

\newlist{steps}{enumerate}{1}
\setlist[steps, 1]{label = \emph{Step} \arabic*:}

\theoremstyle{plain}
\newtheorem{theorem}{Theorem}[section]
\newtheorem{corollary}[theorem]{Corollary}
\newtheorem{lemma}[theorem]{Lemma}
 
\newtheorem{remark}[theorem]{Remark}
\newtheorem{assumption}[theorem]{Assumption}
\newtheorem{proposition}[theorem]{Proposition}
\newtheorem{definition}[theorem]{Definition}
\theoremstyle{remark}
\numberwithin{equation}{section}

\begin{document}		
	\title{}
\title{Price Impact on Term Structure}
\author[]{Damiano Brigo\thanks{Corresponding author. {\tt damiano.brigo@imperial.ac.uk} \ Department of Mathematics, Imperial College London, South Kensington Campus, Weeks Hall, 16--18 Princes Gardens, London.}}
\author[]{Federico Graceffa}
\author[]{Eyal Neuman}
\affil[]{Department of Mathematics, Imperial College London}
\maketitle

\begin{abstract}
We introduce a first theory of price impact in presence of an interest-rates term structure. We explain how one can formulate instantaneous and transient price impact on zero-coupon bonds with different maturities, including a cross price impact that is endogenous to the term structure. We connect the introduced impact to classic no-arbitrage theory for interest rate markets, showing that impact can be embedded in the pricing measure and that no-arbitrage can be preserved. We extend the price impact setup to coupon-bearing bonds and further show how to implement price impact in a HJM framework. We present pricing examples in presence of price impact and numerical examples of how impact changes the shape of the term structure. Finally, we show that our approach is applicable by solving an optimal execution problem in interest rate markets with the type of price impact we developed in the paper.   
\end{abstract}

\bigskip


\bigskip
	
%\tableofcontents

%%%%%%%%%%%%%%%%%%%%%%%%%%%%%%%%%%%%%%%%%%%%%%%%%%%%%
%					Introduction 
%%%%%%%%%%%%%%%%%%%%%%%%%%%%%%%%%%%%%%%%%%%%%%%%%%%%%
\section{Introduction}
The main aim of this work is to present a combined theory of the term structure of interest rates and of price impact, with applications to optimal execution. This objective entails the inclusion of a specific type of cross price impact that is specific to fixed income. 

Term structure modeling with a view to derivatives valuation and hedging has been developed over several decades. For the purposes of this work and what we might call the classic theory we refer to the monographs by Bjork \cite{bjork1997interest}, Brigo and Mercurio \cite{brigo2007interest} and Filipovic \cite{filipovic2009term}. After the crisis that started in 2007, the gap between two of the rates that were used as benchmarks for risk free rates, namely interbank and overnight rates, widened considerably, peaking in October 2008, following the defaults of several financial institutions in the space of one month \cite{brigopallatorre}. This highlighted the fact that interbank rates could no longer be used as benchmarks for risk free rates and neither could they be used to derive zero-coupon curves that were not contaminated by credit and liquidity risk. This led to the necessity to model multiple interest rate curves, treating interbank rates as risky rates affected by credit and liquidity risk and adopting overnight based rates as new risk free rates. 

This multiple interest rate curve academic literature was initiated by practitioners, see in particular the monograph by Henrard \cite{henrard2014interest}. Substantial contributions were made later by numerous academics, where we refer to a monograph by Grbac and Runggaldier \cite{grbac2015interest}, Crepey et al. \cite{crepey2015levy}, Grbac et al. \cite{grbac2015affine}, Cuchiero et al. \cite{cuchiero2019affine,cuchiero2016general}, Nguyen and Seifried \cite{nguyen2015multi} and finally to Bormetti et al. \cite{bormetti2018}, for multiple curves in conjunction with valuation adjustments and credit risk.

Further recent developments include the presence of negative interest rates in many currencies, see for example the BIS report  \cite{negativeratesbis}, and the ongoing project of eliminating current interbank rates like the London Interbank Offer Rate (LIBOR), replacing them with new types of risk free rates inspired by overnight rates. This puts the multiple-curve area in a state of uncertainty, while negative rates prompted the mainstream resurgence of Gaussian models that were previously justified only in very special economies exhibiting negative rates, such as for example Switzerland in the seventies. 

Given the state of market uncertainty on benchmark interest rates, products and markets, we will not consider these recent developments in this work, except for allowing for negative rates in our formulation. We are interested in developing a combination of term structure modeling and price impact in the classic theory of interest rates. We are confident that multiple curves and further discussion of negative rates, if still present in the market after reforms and updated central bank policies, can be incorporated in further work after the classic theory has been developed.

Despite the fact that the bond market size is considerably larger than the equity market size, relevance of liquidity risk in the context of bonds was pointed out in several papers. A report from the Federal Reserve Bank of New York from 2003 \cite{chordia2005empirical}, for instance, showed that bond and stock markets have common factors driving their liquidity. A strong relationship between liquidity in the Treasury bond market and in the stock market was later highlighted also by Goyenko and Ukhov \cite{goyenko2009stock}. Another significant contribution appears in a recent paper by Schneider and Lillo \cite{schneider2019cross}, which studied cross price impact among sovereign bonds.

The financial crisis of 2008-2009 have also raised concerns about the inventories kept by intermediaries. % (mainly investment banks). 
Regulators and policy makers took advantage of two main regulatory changes (Reg NMS in the US and MiFID in Europe) and enforced more transparency on the transactions and hence on market participants positions, which pushed the trading processes toward electronic platforms \cite{citeulike:12047995}.
Simultaneously, consumers and producers of financial products asked for less complexity and more transparency.

This tremendous pressure on the business habits of the financial system, shifted it from a customized and high margins industry, in which intermediaries could keep large (and potentially risky) inventories, to a mass market industry where logistics have a central role. As a result, investment banks nowadays unwind their risks as fast as possible.
 In the context of small margins and high velocity of position changes, trading costs  are of paramount importance. A major factor of the trading costs is the price impact: the faster the trading rate, the more the buying or selling pressure will move the price in a detrimental way. 
 
Academic efforts to quantify and reduce the transaction costs of large trades trace back to the seminal papers of Almgren and Chriss \cite{OPTEXECAC00} and Bertsimas and Lo \cite{BLA98}. In both models one large market participant (for instance an asset manager or a bank) would like to buy or sell a large amount of shares or contracts during a specified duration. %In practice, this reference duration is typically a trading day.
The cost minimization problem turned out to be quite involved, due to multiple constraints on the trading strategies.
On one hand, the price impact demands to trade slowly, or at least at a pace which takes into account the available liquidity (see \cite{citeulike:13497373} and references therein).
On the other hand, traders have an incentive to trade rapidly, because they do not want to carry the risk of an adverse price move far away from their decision price. The importance of optimal trading in the industry generated a lot of variations for the initial mean-variance minimization of the trading costs (see \cite{citeulike:12047995,cartea15book,olivier16book} for details). These type of problems are usually formulated as optimisation problems in the context of stochastic control where the agent tries to minimize the transaction costs which result by the price impact and to reduce the risk associated with holding the assets for too long (see e.g. \cite{citeulike:5282787,GLFT, citeulike:13675369}). 

In this spirit, we will follow the term structure theory as developed by Bjork \cite{bjork1997interest}, modifying it to allow for the inclusion of price impact.  We will start from the simplest possible dynamics, namely one factor short rate models, and extend it later to instantaneous forward rate models. We will introduce price impact formulated on zero-coupon bonds, since they are a possible choice of building blocks for the term structure. We will later connect this with impact on coupon bearing bonds that are more commonly traded. 

We assume that an agent who is executing a large order of bonds is creating two types of price impact which are extensively used in the literature. The first one is an instantaneous (or temporary) price impact, which affects the asset  price only while trading, and fades away immediately after. This type of price impact occurs due to the fact that a large buy (sell) trade consumes the liquidity which is available in the market by ``walking through'' the first few price levels of the limit order book (see e.g. \cite{OPTEXECAC00}  and Chapter 6.3 of \cite{cartea15book}). However, empirical studies have shown that price impact also has a \emph{transient} effect. A short of liquidity due to a large trade creates an imbalance between supply and demand, which in turn pushes the price in a detrimental direction. This effect decays within a short time period after each trade (see \cite{citeulike:13497373} and \cite{Ob-Wan2005}). Execution in presence of transient price impact was studied extensively in the context of optimal control problems (see e.g.  \cite{GSS, GSS2, AFS2,neuman2020optimal,stat-adp18}).  In this work we incorporate these two types of price impact models into a bonds trading framework, as was done in \cite{neuman2020optimal} for equities. 

Applying these price impact models to the term structure will be challenging. At every point in time the term structure of interest rates is a high dimensional object, or even an infinite dimensional one when considering all possible maturities for interest rates or zero-coupon bonds at a given time. This is a unique feature of term structure modeling, where differently from FX or equity modeling for example, we model a whole curve dynamics rather than a point dynamics. We can expect that trading a bond with a specific maturity may impact the price of bonds with different maturities {\emph{on the same currency curve}}. In this sense, the cross price impact is endogenous to the same underlying asset, differently from what happens in other markets.  
%Further cross price impact could be, for example, impact of bonds traded in one currency on bonds in another currency, but we will not address this in the current work, as we will consider a single interest rate curve.  
We will investigate how price impact interacts with no-arbitrage dynamics, and we will encapsulate the effect of price impact in the definition of a new no-arbitrage pricing measure embedding impact itself.  This will be done by extending the market price of risk to an impacted version embedding the bond price impact speed. The impacted zero-coupon bond dynamics will then be written as the unimpacted bond dynamics but under a different measure. We will also introduce an impacted physical measure that could be useful for risk management and risk analysis. Finally, we will illustrate our theory by proposing an application to optimal execution. 

We will not limit ourselves to short rate models. We will also see how in the Heath-Jarrow-Morton model the no-arbitrage drift condition for instantaneous forward rate dynamics can be maintained under price impact by resorting to the modified pricing measure.

%Eyal to go more specific on this.

The paper is structured as follow. 
% I would unpack section 2 in more sections. 
In section \ref{sec:road} we propose a roadmap for readers who are not familiar with at least one of the two areas of the paper, namely term structure modeling and market impact.  
In Section \ref{sec:main} we introduce the short rate models setup and the main theoretical results. In particular, we introduce price impact for zero-coupon bonds, and we look at the impacted market price of risk, absence of arbitrage and the impacted risk neutral measure. We define the impacted yield curve and extend impact to coupon bearing bonds. We further show how to use the impact setup in a HJM framework.  
Section \ref{sec:examples} features a few examples including valuation of impacted Eurodollar futures with the Hull and White model.
Section \ref{sec:numerical_results}	
presents some numerical results illustrating how the yield curve behaves under impact.
Finally, we introduce a result on optimal execution with impacted bonds in Section \ref{sec:optexec}.
Sections \ref{sec:proofs} and \ref{sec:proof_linear_feedback} include proofs that have not been included in the main text. 

%
%  Add conclusions section?
%

%%%%%%%%%%%%%%%%%%%%%%%%%%%%%%%%%%%%%%%%%%%%%%%%%%%%%
%				Roadmap
%%%%%%%%%%%%%%%%%%%%%%%%%%%%%%%%%%%%%%%%%%%%%%%%%%%%%
\section{Roadmap: how to read this paper}\label{sec:road}

This paper combines areas that are typically disjointed in the literature, such as term structure modeling and price impact/optimal execution. It is possible that most readers will be familiar with one area but not the other one. We expect to find few readers who may read the whole paper without effort. We therefore provide a roadmap for the reader who wishes to have a quick understanding of what the paper deals with and how it is organized, with a particular focus on where finding what and what needs to be read more carefully depending on the reader's background. 

{\bf Main section of the paper.} 
Section \ref{sec:main} is the section where both the term structure model is introduced and  the impact model is postulated and developed in detail, linking the developments to no-arbitrage pricing. The reader should start ideally from here. 

\medskip

{\bf Interest rate and bond dynamics under different measures.} In Section 
\ref{subsec:def} we introduce our initial assumption on the interest rate dynamics, assuming it is a one-dimensional short rate model dynamics. This will be relaxed later with a Heath-Jarrow-Morton setting in Section \ref{subsec:HJM} but the paper arguments and innovations are best appreciated in a short-rate setting. We start postulating a diffusion dynamics for the short rate $r$ under the physical measure, giving then also the dynamics of the bond price $P$ under the same measure. We introduce the market price of risk leading to the risk neutral measure dynamics for the zero-coupon bond and show that, as classically done in Bjork with no-arbitrage assumptions related to self-financing and locally risk free portfolios, it does not depend on the bond maturity. With this market price of risk we then derive the dynamics of the short rate $r$ under the risk neutral measure. All this is standard theory that the term structure expert will find immediate and may want to read quickly. For readers who are more into price impact but relatively unfamiliar with term structure modeling, we recommend a careful reading of this introductory part. 

\medskip

{\bf Inventory and impact.} We define the trader's inventory $X$, the execution speed $v$ and the instantaneous and transient impacts, leading to the impacted bond price $\tilde{P}$ in Section \ref{subsec:def}. This part will be easier to the market microstructure expert but, given that applications of price impact to full term structure modeling are not known, we recommend to read carefully this part  given the nuances involved in developing an impact theory under term structure.
We then show how it is possible to incorporate the effect of price impact in a new definition of impacted market price of risk, leading to an impacted risk neutral measure.  
 
{\bf Impacted risk neutral and physical probability measures.} 
This measure is introduced in Section \ref{subsec:rnm}. Under this measure, the impacted zero-coupon bond dynamics is the same as the risk neutral dynamics of the unimpacted bond provided one replaces the Brownian motion under the original risk neutral measure with a Brownian motion under the new impacted measure. The impacted risk neutral measure allows us to define an impacted physical measure as well, and again the bond price dynamics under the impacted physical measure is the same as the bond price dynamics under the original physical measure, provided that the Brownian motion is replaced with a Brownian motion under the new measure. In this setting, we can show that the impacted zero-coupon bond is the impacted risk neutral expectation of the stochastic discount factor. In a way, the impacted risk neutral measure behaves indeed as much as possible as a risk neutral measure in presence of impact. 

One of the main purposes of introducing the risk neutral measure in classic derivatives pricing is the idea that the existence of this measure implies absence of arbitrage. The reader might then wonder whether we can say the same of the impacted risk measure. This is dealt with in the Section \ref{subsec:pricing}.

\medskip

{\bf Derivatives pricing and absence of arbitrage.} First, we show that the existence of an impacted equivalent pricing measure implies that the impacted model is arbitrage free. It then follows that we can extend our result for zero-coupon bonds, namely we can price interest-rate derivatives under impact by taking the impacted risk neutral expectation of their discounted cash flows. We present an example involving Eurodollar futures. 

\medskip

An important part of the theory of price impact for interest rates is what we might call {\bf endogenous cross-impact}. With this we point to the impact of a zero-coupon bond with a given maturity on the price of a zero-coupon bond with a different maturity {\emph{on the same currency curve}}. Exogenous cross impact would be the impact of a bond in one curve on the bond on another curve, but this is not what we deal with here. The theory of endogenous cross impact is developed in Section \ref{subsec:yield}. This has to be read very carefully by all readers, as it represents a very peculiar aspect of any theory of price impact applied to a term structure. In term structure modeling the underlying of derivatives is a whole curve, so price impact has to be postulated also for one part of the curve over other parts. This is not a problem in markets where assets are points rather than curves, like stocks or FX rates. 

\medskip

{\bf Coupon bonds.} The theory developed up to this point of the paper deals with zero-coupon bonds. However, most traded bonds are coupon bonds, so it is necessary to develop a theory of price impact for coupon bearing bonds. This is done in Section \ref{subsec:coupon}. First, by looking at prices as expectations of discounted cash flows, we deduce that the price of impacted coupon bearing bond is simply obtained by the unimpacted one, by replacing all unimpacted zero-coupon bonds in its expression with impacted ones. 
We find a way to connect the impact on the coupon bearing bond to the impact of single zero-coupon bonds, and find assumptions under which the impacted dynamics may even be the same.

\medskip

Finally, as explained earlier, Section \ref{subsec:HJM} goes into quite some detail to show that the theory developed so far can be equally formulated in the {\bf Heath-Jarrow-Morton (HJM) framework}, using instantaneous forward rates to model the term structure. It is well known to interest rate experts that the no-arbitrage condition in the HJM framework is expressed by requiring that the drift in the forward rate dynamics is a specific transformation of the dynamics volatility. We show that the same relationship between drift and volatility holds for impacted instantaneous forward rates under the impact-adjusted risk neutral measure. The classic HJM theory then carries over to the impacted case. This part may be helpful to readers who wish to see our framework work outside the short rate models setting, but if the reader is not too interested in the specific interest rate model used, she can skip this part and still get all the ideas on impacted term structure modeling.

\bigskip

{\bf Pricing examples.} These are presented in Section \ref{sec:examples}, including pricing impacted Eurodollar Futures with Gaussian short rate models. This section is interesting for readers who wish to see how the model can be effectively applied to price interest rate derivatives under price impact.

\bigskip

{\bf Numerical results.} Section \ref{sec:numerical_results} presents numerical examples of impacted yield curves, showing how the term structure changes once price impact is included. We analyze the interplay between the bond pull-to-par and cross-price impact parameters. This can be a good section for readers who wish to develop some intuition on how price impact may change the shape of the term structure.  

\bigskip

{\bf Optimal execution with impacted bonds.} For the reader wondering how this theory of price impact can be included in classic problems such as optimal execution, we solve the optimal execution problem in a specific setting. This is presented in Section \ref{sec:optexec}. The example is rather technical and the solution involves a system of forward-backward SDEs, showing that optimal execution with term structure models can be technically demanding, and this is recommended only for readers who are interested in optimal execution. 

\bigskip

Concluding, the best way to read the paper is sequentially, skipping the parts that are not of interest and looking more carefully at areas outside the specific expertise of the reader. We tried to provide a guide to what can be skipped and on the different emphasis for different areas of expertise above.

%%%%%%%%%%%%%%%%%%%%%%%%%%%%%%%%%%%%%%%%%%%%%%%%%%%%%
%				Model setup and main results
%%%%%%%%%%%%%%%%%%%%%%%%%%%%%%%%%%%%%%%%%%%%%%%%%%%%%
\section{Model setup and main results}\label{sec:main}

%---------------------------------------------------
%			Impacted market price of risk
%---------------------------------------------------
\subsection{Impacted market price of risk, impacted risk neutral measure and absence of arbitrage} \label{subsec:def}
We introduce our initial assumption on the interest rate dynamics, assuming it is a one-dimensional short rate model dynamics. This will be relaxed later with a Heath-Jarrow-Morton setting in Section \ref{subsec:HJM}.

Let us fix a maturity $T>0$ and let $(\Omega, \mathcal{F}, (\mathcal{F}_t)_{t \geq 0}, \mathbb{P})$ be a filtered probability space satisfying the usual conditions, on which there is a standard $(\mathcal{F}_t)_{t \in [0,T]}$-Brownian motion $W^\mathbb{P}$. We consider the following dynamics of the short rate under the real world measure $\mathbb{P}$
\begin{equation} \label{r-sde} 
d r(t) = \mu(t,r(t)) dt + \sigma(t,r(t)) d W^\mathbb{P}(t),
\end{equation}
where $\mu(t,r), \sigma(t,r)$ are given real valued functions, assumed to be regular enough to ensure the SDE has a unique strong solution. For example one can assume that both $\mu$ and $
\sigma$ are Lipschitz continuous in the $r$ coordinate, and has at most linear growth in $r$ uniformly in $t\in [0,T]$. 
We moreover assume that $\sigma(t,r(t))$ is $\mathbb{P}$-a.s.  strictly positive for any $t>0$.  

Assume that the dynamics of a zero-coupon bond with maturity at time $T$, under the real world measure, is given by
\begin{equation} \label{class-b}
d P(t,T) = \mu_T(t,r(t)) dt + \sigma_T(t,r(t)) d W^\mathbb{P}(t),
\end{equation}
with $\mu_T,\sigma_T$ depending on the maturity $T$ and regular enough as in \eqref{r-sde}. Typically, one might assume the price process of the $T$-bond to be of the form $P(t,T) = F(t,r(t);T)$ for some function $F$ smooth in three variables. Then, under suitable assumptions (see Assumption 3.2 in Chapter 3.2 of \cite{bjork1997interest}) one can define for any finite maturity $T>0$ the stochastic process
\begin{equation}
\lambda(t) = \frac{\mu_T(t,r(t)) - r(t) P(t,T)}{\sigma_T(t,r(t))},
\label{market_price_of_risk}
\end{equation}
and show that $\lambda$ may depend on $r$ but it does not actually depend on $T$. Such process is called \emph{market price of risk}. Provided the Novikov condition holds, this process can be used to define a change of measure from the real world measure $\mathbb{P}$ to the risk neutral measure $\mathbb{Q}$:
\begin{equation}
\frac{d \mathbb{Q}}{d \mathbb{P}} = \exp \left( \int_0^t \lambda(s) d W^\mathbb{P}(t)- \frac{1}{2} \int_0^t \lambda^2(s) ds\right).
\label{risk_neutral_measure_Q}
\end{equation}
The dynamics of the short rate under $\mathbb{Q}$ becomes
\begin{equation*}
d r(t) = [\mu(t,r(t)) - \lambda(t) \sigma(t,r(t))] dt + \sigma(t,r(t)) d W^\mathbb{Q}(t).
\end{equation*}
%Due to the lack of completeness of the bond market, the model will be specified once the stochastic process $\lambda$ is defined. 
The model will be fully specified once the stochastic process $\lambda$ is defined.
Our strategy for establishing a mathematical framework that encompasses both risk neutral pricing and price impact in the context of interest rates derivatives consists, first of all, in specifying the dynamics for an impacted bond with maturity $T$ under the real world measure $\mathbb{P}$. 

We consider a trader with an initial position of $x_T>0$ zero-coupon bonds with maturity $T$. Let $0< \tau \leq T$ denote some finite deterministic time horizon. In an optimal execution problem, the objective of the trader would be to complete her transaction by time $\tau$, starting from the $x_T$ position at time $0$. In this sense, we should avoid confusion between $T$, which is the traded bond maturity, and $\tau$, which is the trading horizon of the $T$-maturity bond. The number of bonds the trader holds at time $t \in [0,\tau]$ is given by
\begin{equation}  \label{inv} 
X_T(t) = x_T- \int_{0}^t v_T(s) ds.
\end{equation} 
where the function $v_T$ denotes the trader's selling rate, which takes negative values in case of a buy strategy. In what follows we assume that $v_T= \{v_T(t)\}_{0\leq t\leq \tau}$ is progressively measurable and has a $\mathbb P$-a.s. bounded derivative (in the $t$-variable), that is, there exists $M>0$ such that 
 \begin{equation} \label{v-bnd} 
 \sup_{0\leq t\leq \tau^{+}} |\partial_{t}  v_{T}(t)|  < M, \quad \mathbb P-\rm{a.s.}, 
\end{equation} 
where $0\leq t\leq \tau$. After the trading stops, we assume that $v_{T}(t)= 0$.
We denote the class of such trading speeds as $\mathcal{A}_{T}$.  
 
The main idea behind the assumption of the differentiability of $v_{T}$ is that the overall impact we add to the zero-coupon bond should affect the drift only (see \eqref{impacted_bond_differential_form}). Moreover due to price impact effect, we allow  bond price which are large than $1$ for some time intervals but we do need to control their upper bound. 

We consider a price impact model with both transient and instantaneous impact, which is a slight generalization of the model which was considered in \cite{neuman2020optimal}. The impacted bond price is therefore given by 
\begin{equation} \label{impacted_bond}
\tilde{P}(t,T) = P(t,T) - l(t,T)  v_T(t) - K(t,T) \Upsilon_T^v(t). 
\end{equation}
Here $\Upsilon_T^v$ represents the transient impact effect and it has the form 
\begin{equation}\label{def:transient_impact}
\Upsilon_T^v(t):= y e^{-\rho t} + \gamma \int_0^t e^{-\rho(t-s)} v_T(s) ds,
\end{equation}
where $y,\rho$ and $\gamma$ are positive constants. The term $v_T(t)$ in \eqref{impacted_bond} represents the instantaneous price impact, where we absorb in the function $l$ any constants that should factor it. 
Lastly, $l$ and $K$ are differentiable functions with respect to both variables $(t,T)$ which take positive values on $0\leq t <T$ and for any $0\leq \tau <T$ we have
\begin{equation}  \label{l-pos}
\inf_{0\leq t \leq \tau} l(t,T) >0.
\end{equation} 
Moreover we assume that   
\begin{equation}  \label{k-assump} 
\begin{aligned} 
&\sup_{0\leq t\leq \tau} |\partial_{t}l(t,T)| <\infty, \quad \lim_{t \rightarrow T} l(t,T) =0,  \\
&    \sup_{0\leq t\leq \tau} |\partial_{t}K(t,T)| <\infty, \quad \lim_{t \rightarrow T} K(t,T) =0.
\end{aligned} 
\end{equation}  
While the assumption on boundedness of the derivatives of functions $K$ and $l$ arise from technical reasons which has similar motivation as the reason for \eqref{v-bnd}, the assumptions on the behaviour at expiration is meant to enforce the boundary condition on the price of the impacted bond at expiration, which is $\tilde{P}(T,T) =1$. Note that $K$ and $l$ are time-dependent versions of the parameters $\lambda,k$ in \cite{neuman2020optimal}. A prominent example of such functions is 
\[
l(t,T) = \kappa \left(1-\frac{t}{T} \right)^\alpha, \quad K(t,T) = \left(1-\frac{t}{T} \right)^\beta, 
\]
for some constants $\alpha , \beta \geq1$ and $\kappa>0$. 

We define for convenience the overall price impact:
\begin{equation}\label{def:overall_price_impact}
I_T(t) := l(t,T) v_T(t) + K(t,T) \Upsilon_T^v(t).
\end{equation}
Then, since $v_{T}$ is in $\mathcal A_{T}$ we can rewrite \eqref{impacted_bond} as follows: 
\begin{equation} \label{impacted_bond_differential_form}
d \tilde{P}(t,T) = d P(t,T) - J_T(t) dt, \quad \tilde P(T,T) = 1,
\end{equation}
with
\begin{equation} \label{impact_density}
\begin{aligned}
J_T(t) & := \partial_t I_T(t) \\
& = \partial_t l(t,T) v_T(t) + l(t,T) \partial_t v_T(t) + \partial_t K(t,T) \Upsilon_T^v(t) + K(t,T) [-\rho \Upsilon_T^v(t) + v_T(t)].
\end{aligned}
\end{equation} 

Our model so far describes how trading a $T$-bond affects its price. Next, we show the existence of an \emph{impacted market price of risk process} which will be a generalization of \eqref{market_price_of_risk}. Using this process we will define an equivalent martingale measure, under which bonds and derivatives prices can be computed. Such a measure will be called an \emph{impacted risk neutral measure}. It is important to remark that, as in the classic case, this change of measure will be unique for all bond maturities. 

Before stating the main theorem of this section, let us first introduce a few important definitions.

%----------------------- Definition: impacted portfolio -------------------%
\begin{definition}[Impacted portfolio]\label{def:impacted_portfolio}
Let $\hat{T}<+\infty$ be some finite time horizon. An \emph{impacted portfolio} is a $(n+1)-$dimensional, bounded progressively measurable process $\tilde{h} = (\tilde{h}_t)_{t \in [0, \hat T]}$ with $\tilde{h}_t = (\tilde{h}_t^0,\tilde{h}_t^1,\dots,\tilde{h}_t^n)$, where $\tilde{h}_t^i$ represents the number of shares in the impacted bond $\tilde{P}(t,T_i)$ held in the portfolio at time $t$. The value at time $t$ of such a portfolio $\tilde{h}$ is defined as
\begin{equation*}
\tilde{V}(t) \equiv \tilde{V}(t,\tilde{h}) := \sum_{i=0}^n \tilde{h}^i(t) \tilde{P}(t,T_i).
\end{equation*}
\end{definition}

%----------------------- Definition: self-financing -------------------%
\begin{definition}[Self-financing]\label{def:self_financing}
Let $\hat{T}<+\infty$ be some finite time horizon. and let $\tilde{h}$ be an impacted portfolio as in Definition \ref{def:impacted_portfolio}. We say that $\tilde{h}$ is \emph{self-financing} if its value $\tilde{V}$ is such that
\begin{equation} \label{imp-port} 
d \tilde{V}(t,\tilde{h}) = \sum_{i=0}^n \tilde{h}^i(t) d  \tilde P(t,T_i), \quad \textrm{for all } 0\leq t \leq \hat T. 
\end{equation}
\end{definition}

%----------------------- Definition: locally risk free ---------------------%
\begin{definition}[Locally risk free]\label{def:locally_risk_free}
Let $\tilde{h}$ be an impacted portfolio as in Definition \ref{def:impacted_portfolio} and let $\tilde{V}$ be its value. Let also $\alpha=(\alpha_t)_{t \in [0,\hat T]}$ be an adapted process. We say that $\tilde{h}$ is \emph{locally risk-free} if, for almost all $t$,
\begin{equation*}
d \tilde{V}(t) = \alpha(t) \tilde{V}(t) \implies \alpha(t) = r(t),
\end{equation*}
where $r(t)$ is the risk-free interest rate introduced in \eqref{r-sde}.
\end{definition}

Here is the main result of this section.

%---------------------- Theorem: Impacted market price of risk ---------------%
\begin{theorem}[Impacted market price of risk]\label{thm_impacted_market_price}
Let $\hat{T}<+\infty$ be some finite time horizon and let $\mathbb T:=(0,\hat{T}]$.
Let $J_T$ be the impact density defined in \eqref{impact_density}. Given an impacted portfolio $\tilde{h}$ as in \eqref{imp-port}, we assume that it is self-financing and locally risk-free, as in Definitions \eqref{def:self_financing} and \eqref{def:locally_risk_free}, respectively. Then, there exists a progressively measurable stochastic process $\tilde{\lambda}(t)$ such that 
\begin{equation} \label{def:lambda_tilde}
\tilde{\lambda}(t) = \frac{\mu_{T_i}(t,r(t)) - r(t) \tilde{P}(t,T_i) - J_{T_i}(t)}{\sigma_{T_i}(t,r(t))}, \quad t\geq 0, 
\end{equation}
for each maturity $T_i $, $i=1,..,n$, with $\tilde{\lambda}$  depending on the short rate $r$ but not on $T_i$.
 \end{theorem}
 The proof of Theorem \ref{thm_impacted_market_price} is given in Section \ref{sec:proofs}.

 %-------------- Remark:  Self-financing in presence of price impact -------------------%
 \begin{remark}[Self-financing in presence of price impact]
In presence of price impact it is of course not obvious that the self-financing condition should hold. Adjusted self-financing conditions have been proposed, for instance, by Carmona and Webster \cite{carmona2013self}. We notice that, in their work, the adjustment consists of two parts: the covariation between the inventory and the price process, and the bid-ask spread. In our work we will assume the inventory is a finite variation process and that the bid ask spread is negligible, thereby obtaining the classic self-financing condition.
\end{remark}

%----------------------- Remark:  Intrinsic price impact------------------------%
\begin{remark} [Intrinsic price impact] 
From Theorem \ref{thm_impacted_market_price} it follows that \emph{endogenous cross price impact} naturally emerges in our framework. Indeed, once an agent trades a bond with maturity $T_1$, the process $\tilde{\lambda}$ is uniquely determined. Note that $\tilde \lambda$ does not depend on the maturity. For any bond with maturity $T_2 \in \mathbb{T}$, which is not traded, we have $J_{T_2}\equiv 0$ but by \eqref{def:lambda_tilde}, the price $\tilde P(t,T_2)$ will be affected by the trade on the bond with maturity $T_1$. We remark that, by \emph{endogenous}, we mean that the bonds with different maturities $T_1$ and $T_2$ are thought of as belonging to the \emph{same currency curve}. If we were to discuss multiple interest rate curves, then exogenous cross price impact should be taken into account as well.
\end{remark} 

%---------------------------------------------------
%			Impacted risk-neutral measure
%---------------------------------------------------
\subsection{Impacted risk-neutral measure}\label{subsec:rnm}
We previously introduced two measures: the real world measure $\mathbb{P}$ and the classic risk neutral measure $\mathbb{Q}$, as defined in \eqref{risk_neutral_measure_Q}. Now we use the result of Theorem \ref{thm_impacted_market_price} to define a third measure, which we call \emph{impacted risk neutral measure} and denote by $\tilde{\mathbb{Q}}$. This is defined as follows:
\begin{equation}   \label{radon_nikodym_der_Q_tilde}
\frac{d \tilde{\mathbb{Q}}}{d \mathbb{P}} = \exp \left\{ \int_0^t \tilde{\lambda}(s) d W^\mathbb{P}(s) - \frac{1}{2} \int_0^t \tilde{\lambda}^2(s) ds \right\}.
\end{equation}
The well posedness of $\tilde{\mathbb{Q}}$ can be checked via the Novikov condition. It might be useful to recall that the usual approach does not consist in determining the conditions on $\mu_T,\sigma_T$ under which the Novikov condition is fulfilled. Rather, one chooses a specific short rate model to begin with. Then, one can specify the market price of risk process, exploiting the fact that it depends on $t$ and $r$, but not on $T$. For example, in the case of Vasicek model, the market price of risk is assumed to be $\lambda(t) = \lambda r(t)$, for some constant $\lambda$. At this point, Novikov condition can be checked much more easily. Since we proved that $\tilde{\lambda}$ depends on $t$ and $r$ only, we can assume the two processes to have the same structure and follow the same idea. In the case of Vasicek model, for example, we can assume $\tilde{\lambda}(t) = \tilde{\lambda} r(t)$, for some constant $\tilde{\lambda}$ incorporating the impact. Consequently, determining the existence and well-posedness of $\tilde{\mathbb{Q}}$ is fundamentally equivalent to determining the existence and well-posedness of $\mathbb{Q}$.

The Girsanov change of measure from the classic risk neutral measure to the impacted one is given by
\begin{equation*}
\frac{d \tilde{\mathbb{Q}}}{d \mathbb{Q}} = \frac{d \tilde{\mathbb{Q}}}{d \mathbb{P}} \frac{d \mathbb{P}}{d \mathbb{Q}}.
\end{equation*}
with
\begin{equation*}
\frac{d \mathbb{P}}{d \mathbb{Q}} = \exp \left\{- \int_0^t \lambda(s) d W^\mathbb{Q}(s) - \frac{1}{2} \int_0^t \lambda^2(s) ds \right\}.
\end{equation*}
where $\lambda$ was defined in \eqref{market_price_of_risk}.
Hence,
\begin{equation*}
\frac{d \tilde{\mathbb{Q}}}{d \mathbb{Q}} = \exp \left\{ \int_0^t \tilde{\lambda}(s) d W^\mathbb{P}(s) - \frac{1}{2} \int_0^t \tilde{\lambda}^2(s)ds + \int_0^t \lambda(s) d W^\mathbb{Q}(s) - \frac{1}{2} \int_0^t \lambda^2(s)ds \right\}.
\end{equation*}
Since $W^\mathbb{P}(t) := W^\mathbb{Q}(t) + \int_0^t \lambda(s) ds$ is a Brownian motion under the measure $\mathbb{P}$, we have
\begin{equation*}
\frac{d \tilde{\mathbb{Q}}}{d \mathbb{Q}} = \exp \left\{ \int_0^t (\tilde{\lambda}(s) - \lambda(s)) d W^\mathbb{Q}(s) - \frac{1}{2} \int_0^t \left(\lambda^2(s) + \tilde{\lambda}^2(s) - 2 \lambda(s) \tilde{\lambda}(s) \right) ds \right\}.
\end{equation*}
In other words,
\begin{equation*}
W^{\tilde{\mathbb{Q}}}(t) := W^\mathbb{Q}(t) - \int_0^t (\tilde{\lambda}(s) - \lambda(s)) ds,
\end{equation*}
is a Brownian motion under the measure $\tilde{\mathbb{Q}}$. It is then straightforward to notice that the impacted zero-coupon bond under the impacted measure $\tilde{\mathbb{Q}}$ will be described by the dynamics
\begin{equation} \label{qt-p}
d \tilde{P}(t,T) = r(t) \tilde{P}(t,T) dt + \sigma_T(t,r(t)) dW^{\tilde{\mathbb{Q}}}(t).
\end{equation}
We further remark that, in principle, we could start by defining a new measure $\tilde{\mathbb{P}}$ to get rid of the additional drift due to impact. Just rewrite the dynamics of the impacted zero-coupon bond as
\begin{equation*}
d \tilde{P}(t,T) = \mu_T(t,r(t)) dt + \sigma_T(t,r(t)) \left( \frac{J_T(t)}{\sigma_T(t,r(t))} dt  + d W^{\mathbb{P}}(t) \right).
\end{equation*}
This suggests to define
\begin{equation*}
\frac{d \tilde{\mathbb{P}}}{d \mathbb{P}} = \exp \left\{ \int_0^t \frac{J_T(s)}{\sigma_T(s,r(s))} d W^\mathbb{P}(s) - \frac{1}{2} \int_0^t \left(\frac{J_T(s)}{\sigma_T(s,r(s))} \right)^2 ds \right\}.
\end{equation*}
The impacted bond under this measure would follow the dynamics
\begin{equation}
d \tilde{P}(t,T) = \mu_T(t,r(t)) dt + \sigma_T(t,r(t)) dW^{\tilde{\mathbb{P}}}(t).
\label{impacted_bond_under_impacted_P}
\end{equation}
At this point, $\tilde{\mathbb{Q}}$ can be defined from $\tilde{\mathbb{P}}$ by using the classic market price of risk $\lambda(t)$. In other words,
\begin{equation*}
\frac{d \tilde{\mathbb{Q}}}{d \tilde{\mathbb{P}}} = \frac{d \tilde{\mathbb{Q}}}{d \mathbb{Q}} \frac{d \mathbb{Q}}{d \tilde{\mathbb{P}}} = \frac{d \tilde{\mathbb{P}}}{d \mathbb{P}} \frac{d \mathbb{Q}}{d \tilde{\mathbb{P}}} = \frac{d \mathbb{Q}}{d \mathbb{P}}.
\end{equation*}
Putting everything together, we have the following commuting diagram
\[
\begin{tikzcd}
\mathbb{P} \arrow{r}{\lambda}  \arrow[swap]{dr}{\tilde{\lambda}} \arrow[swap]{d}{\tilde{\lambda}-\lambda} & \mathbb{Q} \arrow{d}{\tilde{\lambda} - \lambda} \\
\tilde{\mathbb{P}} \arrow{r}{\lambda} & \tilde{\mathbb{Q}}
\end{tikzcd}
\]
By \eqref{qt-p} and usual arguments it follows that discounted impacted traded prices, that is $\{\tilde P(\cdot,T)/B(t)\}_{t\geq 0}$, are martingales for any $0\leq T \leq \hat T$ under $\tilde{\mathbb{Q}}$. Here $B$ is the usual money market account at time $t$ given by 
\begin{equation}\label{def:bank_account}
B(t) = e^{\int_0^t r(s) ds}.
\end{equation}
We therefore have 
\begin{equation} \label{discounted_impacted_ZC_bond}
\frac{\tilde{P}(t,T)}{B(t)} = \mathbb{E}^{\tilde{\mathbb{Q}}} \left[ \frac{\tilde{P}(T,T)}{B(T)} \Bigg| \mathcal{F}_t \right].
\end{equation}
Multiplying both sides by $B(t)$ and exploiting the boundary condition $\tilde{P}(T,T)=1$, we obtain the fundamental equation
\begin{equation}
\tilde{P}(t,T) = \mathbb{E}^{\tilde{\mathbb{Q}}} \left[ e^{-\int_t^T r(s) ds} \big | \mathcal{F}_t \right].
\label{impacted_ZCB_expectation}
\end{equation}

%-------- Remark: Interpretation impacted real world measure -------------------%
\begin{remark}[Interpretation impacted real world measure]
From \eqref{impacted_bond_under_impacted_P} we observe that, financially speaking, under the impacted real world measure $\tilde{\mathbb{P}}$, impacted bond dynamics $\tilde P(\cdot,T)$ has the same dynamics as the classic bond (without price impact modeling) in \eqref{class-b} under $\mathbb{P}$. In particular we have that $\tilde P(T,T)=1$ for all maturities $T$.
\end{remark} 

%-----------------------------------------------------------------
%		Applications to pricing of interest rate derivatives
%-----------------------------------------------------------------
\subsection{Applications to pricing of interest rate derivatives} \label{subsec:pricing}

We start this section by remarking that the notion of arbitrage we use in our work is the classic one (see e.g. Harrison and Kreps \cite{harrison1979martingales} or Harrison and Pliska \cite{harrison1981martingales}), adjusted with impacted portfolios.
\begin{definition}[Arbitrage portfolio]
An \emph{arbitrage portfolio} is an impacted self-financing portfolio $\tilde{h}$ such that its corresponding value process $\tilde{V}$ satisfies
\begin{enumerate}
\item $\tilde{V}(0) = 0$ . 
\item $\tilde{V}(T)  \geq 0$ $\mathbb{P}$-a.s.
\item $\mathbb{P}(\tilde{V}(T)>0)>0$
\end{enumerate}
\end{definition}
Using this definition of arbitrage we show that the first fundamental theorem of asset pricing holds in our setting.

%-------------------- Theorem: Absence of arbitrage --------------------------%
\begin{theorem}[Absence of arbitrage]\label{thm_absence_arbitrage}
Assume that there exists an impacted equivalent martingale measure $\tilde{\mathbb{Q}}$ as in \eqref{radon_nikodym_der_Q_tilde}. Then, our impacted model is arbitrage  free.
\end{theorem}
The proof of Theorem \ref{thm_absence_arbitrage} is given in Section \ref{sec:proofs}.

A key consequence Theorem \ref{thm_absence_arbitrage} is that our term structure model with price impact is indeed free of arbitrage. This allows to price interest rate derivatives by taking the expectation of discounted payoffs under the impacted risk neutral measure $\tilde{\mathbb{Q}}$. As a benchmark example, we consider the price of an impacted Eurodollar future. In the classic context, a Eurodollar-futures contract provides its owner with the payoff (see Chapter 13.12 of \cite{brigo2007interest})
\begin{equation*}
N (1-L(S,T)),
\end{equation*}
where $N$ denotes the notional and $L(S,T)$ is the LIBOR rate, defined as (see Chapter 1 of \cite{brigo2007interest}, Definition 1.2.4) 
\begin{equation}
L(S,T) := \frac{1 - P(S,T)}{\tau(S,T) P(S,T)},
\label{libor_rate}
\end{equation}
with $\tau(S,T)$ denoting the year fraction between $S$ and $T$. Motivated by this, we introduce the impacted counterpart of the LIBOR rate in \eqref{libor_rate}, i.e.
\begin{equation} \label{l-t}
\tilde{L}(S,T) := \frac{1 - \tilde{P}(S,T)}{\tau(S,T) \tilde{P}(S,T)},
\end{equation}
with $\tau$ defined as above and the impacted zero-coupon bond in place of the classic one. This new rate $\tilde{L}$ is interpreted as the simply-compounded rate which is consistent with the impacted bond. This corresponds to the classic LIBOR rate which is the constant rate at which one needs to invest $P(t,T)$ units of currency at time $t$ in order to get an amount of one unit of currency at maturity $T$. Then, the fair price of an impacted Eurodollar future at time $t$ is (see \cite{brigo2007interest}, Chapter 13, eq. (13.19))
\begin{equation}\label{fair_price_impacted_Eurodollar}
\begin{split}
\tilde{C}_t & = \mathbb{E}_t^{\tilde{\mathbb{Q}}} [N (1-\tilde{L}(S,T))], \\
& = N \left(1 + \frac{1}{\tau(S,T)} - \frac{1}{\tau(S,T)} \mathbb{E}_t^{\tilde{\mathbb{Q}}} \left[\frac{1}{\tilde{P}(S,T)} \right] \right),
\end{split}
\end{equation}
where the discounting was left out due to continuous rebalancing (see again Chapter 13.12 of \cite{brigo2007interest}). We will demonstrate in Section \ref{sec:examples} how such expectation can be computed analytically provided the short rate model is simple enough as in Vasicek and Hull-White models. 

%---------------- Remark: Linear and nonlinear pricing equations ----------------%
\begin{remark}[Linear and nonlinear pricing equations]
Our success in retaining analytical tractability and linearity in the pricing equation may look surprising at first. In the context of equities, pricing derivatives in presence of price impact typically leads to nonlinear PDEs. This, in turn, motivated the study of super-replicating strategies and the so-called gamma constrained strategies. Several works provide also necessary and sufficient conditions ensuring the parabolicity of the pricing equation, hence the existence and uniqueness of a self-financing, perfectly replicating strategy. We refer, for example, to Abergel and Loeper \cite{abergel2013pricing}, Bourchard, Loeper et al. \cite{bouchard2016almost,bouchard2017hedging} and Loeper \cite{loeper2018option}. The point we would like to stress here is that the nonlinearity of the pricing equation is a consequence of the trading strategy having a diffusion term, or a consequence of the presence of transaction costs. In other words, under the assumption that trading strategies have bounded variation and no transaction costs are present, the pricing PDE becomes linear again. Hence, our work is actually in agreement to what can be found in the context of equities.
\end{remark}

%-------------------------------------------------------
%		Cross price impact and impacted yield curve
%-------------------------------------------------------
\subsection{Cross price impact and impacted yield curve} \label{subsec:yield}
In this section we discuss how trading a bond $P(t,T)$ impacts the yield curve. For the sake of analytical tractability, we will consider affine short-rate models, that is, those models where bond prices are of the form
\begin{equation}
P(t,T) = A(t,T) e^{-B(t,T) r(t)}, \quad 0\leq t \leq T, 
\label{affine:bond_price}
\end{equation}
for some deterministic, smooth functions $A$ and $B$ and $r$ is given by \eqref{r-sde}. The remarkable property of these models is that they can be completely characterized as in the following theorem (see, e.g., Filipovic \cite{filipovic2009term}, Section 5.3, Brigo and Mercurio \cite{brigo2007interest}, Section 3.2.4, Bjork \cite{bjork1997interest} Section 3.4 and references therein). 

%---------------------------- Characterization affine short-rate models --------------%
\begin{lemma}[Characterization affine short-rate models]
The short rate model \eqref{r-sde} is affine if and only if there exist deterministic, continuous functions $a,\alpha,b,\beta$ such that the diffusion and the drift terms in \eqref{r-sde} are of the form
\begin{align*}
\begin{split}
\sigma^2(t,r) & = a(t) + \alpha(t) r, \\
\mu(t,r) & = b(t) + \beta(t) r,
\end{split}
\end{align*}
and the functions $A,B$ satisfy the following system of ODEs
\begin{align*}
\begin{split}
- \frac{\partial}{\partial t}  \ln A(t,T) & = \frac{1}{2} a(t) B^2(t,T) - b(t) B(t,T), \ \ \ A(T,T) = 1, \\
\frac{\partial}{\partial t} B(t,T) & = \frac{1}{2} \alpha(t) B^2(t,T) - \beta(t) B(t,T) - 1, \ \ \ B(T,T) = 0, \\
\end{split}
\end{align*}
for all $t \leq T$.
\end{lemma}
As explained in \cite{filipovic2009term}, the functions $a,\alpha,b,\beta$ can be further specified by observing that any non-degenerate short rate affine model, that is a model with $\sigma(t,r) \ne 0$ for all $t>0$, can be transformed, by means of an affine transformation, in two cases only, depending on whether the state space of the short rate $r$ is the whole real line $\mathbb{R}$ or only the positive part $\mathbb{R}_+$. In the first case, it must hold $\alpha(t)=0$ and $a(t) \geq 0$, with $b,\beta$ arbitrary. In the second case, it must hold $a(t)=0, \alpha(t), b(t) \geq 0$ and $\beta$ arbitrary.

Let $\hat{T}<+\infty$ be some finite time horizon. The yield curve at a pre-trading time $t_0$ (i.e. before price impact effects kick in) according to classic theory of interest rates is defined by 
\begin{equation}
Y(t,T) := P(t,T)^{-1/T} - 1, \quad 0\leq t \leq t_0,
\label{def:classic_yield}
\end{equation}
for all maturities $0 \leq T \leq \hat T$. Next, we consider the impacted bond dynamics 
\begin{equation*}
d \tilde{P}(t,T) = d P(t,T) - J_T(t) dt, 
\end{equation*}
where $J_T$ was defined in \eqref{impact_density}. Recall that the dynamics of $r(t)$ is given in \eqref{r-sde}. Applying Ito's formula on $P(t,T)$ in  \eqref{affine:bond_price} we get
\begin{multline*}
d P(t,T) = e^{-B(t,T) r(t)} \bigg[\frac{\partial A}{\partial t} - A(t,T) \frac{\partial B}{\partial t} r(t) + \frac{1}{2} A(t,T) B^2(t,T) \sigma^2(t,r(t)) + \\
- A(t,T) B(t,T) \mu(t,r(t)) \bigg] dt - \sigma(t,r(t)) B(t,T) A(t,T) e^{-B(t,T) r(t)} d W^{\mathbb{P}}(t). \\
\end{multline*}
From this equation, we readily extract the drift and the diffusion of the zero-coupon bond with maturity $T$: 
\begin{align} \label{affine:drift_vol_zero_coupon_bond}
\begin{split}
\mu_T(t,r(t)) &:= e^{-B(t,T) r(t)} \bigg[\frac{\partial A}{\partial t} - A(t,T) \frac{\partial B}{\partial t} r(t) + \frac{1}{2} A(t,T) B^2(t,T) \sigma^2(t,r(t))  \\
& \quad - A(t,T) B(t,T) \mu(t,r(t)) \bigg], \\
\sigma_T(t,r(t)) & := - \sigma(t,r) A(t,T) B(t,T) e^{-B(t,T) r(t)}.
\end{split}
\end{align}
Next, we consider the effect of an agent trading on the bond with maturity $T$ on a bond which is not traded by the agent with maturity $S$. We call this effect the \emph{endogenous cross-impact} on the bond with maturity $S$. Recall that in this case the dynamics of the $S$-bond is given by
\begin{equation} \label{sde:impacted_bond_S}
d \tilde{P}(t,S) = \mu_S(t,r(t)) dt + \sigma_S(t,r(t)) d W^{\mathbb{P}}(t), 
\end{equation}
where the coefficients $\mu_S$ and $\sigma_S$ are given by analogous formulas to \eqref{affine:drift_vol_zero_coupon_bond}. Since we are trading the $T$-bond only, $J_S$ in \eqref{impact_density} will be identically equal to zero. Hence, the definition of the impacted market price of risk \eqref{def:lambda_tilde} implies the following relationship
\begin{equation*}
\frac{\mu_T(t,r(t)) - r(t) \tilde{P}(t,T) - J_T(t)}{\sigma_T(t,r(t))} = \frac{\mu_S(t,r(t)) - r(t) \tilde{P}(t,S)}{\sigma_S(t,r(t))}.
\end{equation*}
This equation tells us how the drift the $S$-bond has to change in order to avoid arbitrage. That is, this equation describes the \emph{cross-price impact}. Specifically we have
\begin{equation*}
\mu_S(t,r(t)) = \frac{\sigma_S(t,r(t))}{\sigma_T(t,r(t))} \left[\mu_T(t,r(t)) - r(t) \tilde{P}(t,T) - J_T(t)\right] + r(t) \tilde{P}(t,S).
\end{equation*}
Substituting this drift in \eqref{sde:impacted_bond_S} we get 
\begin{align}\label{sde:cross_impacted_bond_S}
\begin{split}
d \tilde{P}(t,S) &= r(t) \tilde{P}(t,S) dt + \frac{\sigma_S(t,r(t))}{\sigma_T(t,r(t))} \left[\mu_T(t,r(t)) - r(t) \tilde{P}(t,T) - J_T(t) \right] dt \\
& \quad + \sigma_S(t,r(t)) d W^{\mathbb{P}}(t). \\
\end{split}
\end{align}
Finally, we define the impacted yield curve for all $t_0 \leq T \leq \hat T$ as follows: 
\begin{equation} \label{def:impacted_yield}
\tilde{Y}(t,T) := \tilde{P}(t, T)^{-1/T} - 1.
\end{equation}

%------------------Remark: Cross impacted bonds at maturity ------------------------------%
\begin{remark}[Cross impacted bonds at maturity]\label{rem:cross_impacted_bonds_at_maturity}
We have shown in \eqref{impacted_bond_differential_form} that according to our model $\tilde P(T,T) =1$. However, we should also ensure that all cross-impacted bonds with maturity $S \not = T$ reach value $1$ at their  maturities. This of course, would make the model much more involved and we may lose tractability. 
\end{remark}

%---------------------------------------------------
%				Coupon bonds
%---------------------------------------------------
\subsection{Coupon bonds}\label{subsec:coupon}
It is worth recalling that the zero-coupon bond $P(t,T)$ is rarely traded. In practice, its price is derived using some bootstrapping procedure applied, for instance, to coupon bonds. In the classic theory, coupon bonds are defined as
\begin{equation*}
B(t,T) = \sum_{i=1}^n c_i P(t,T_i) + N P(t,T_n),
\end{equation*}
where $N$ denotes the reimbursement notional, $(c_i,T_i)_{i=1}^n$ are the coupons and the maturities at which the coupons are paid, respectively. In order to determine an expression for the impacted coupon bond, we start from its cash flow
\begin{equation*}
C(t) := \sum_{i=1}^n c_i D(t,T_i) + N D(t,T_n),
\end{equation*}
where $D(t,T)$ is the stochastic discount factor defined by
\begin{equation*}
D(t,T) := e^{- \int_t^T r(s) ds},
\end{equation*}
where $r$ is given by \eqref{r-sde}.
Then, we define the impacted coupon bond as the expectation of this cash flow under the impacted risk neutral measure $\tilde{\mathbb{Q}}$ (see \eqref{radon_nikodym_der_Q_tilde}):
\begin{equation*}
\tilde{B}(t,T) := \mathbb{E}^{\tilde{\mathbb{Q}}} \left[C(t) \right].
\end{equation*}
Substituting the expression of $C$ immediately yields
\begin{align}\label{impacted_coupon_bond_linear_combination}
\begin{split}
\tilde{B}(t,T) & = \sum_{i=1}^n c_i \mathbb{E}^{\tilde{\mathbb{Q}}} \left[ D(t,T_i) \right] + N \mathbb{E}^{\tilde{\mathbb{Q}}} \left[ D(t,T_n) \right] \\
& = \sum_{i=1}^n c_i \tilde{P}(t,T_i) + N \tilde{P}(t,T_n),
\end{split}
\end{align}
where $\tilde{P}(\cdot, T_i)$ is the (directly) impacted price of a zero-coupon bond as defined in \eqref{impacted_bond}. Note that \eqref{impacted_coupon_bond_linear_combination} gives the price of the impacted coupon bond in terms of impacted zero-coupon bonds. Since zero-coupon bonds are not always traded, we would like to get a direct pricing formula for impacted coupon bonds. Let $\{v_{T_i}\}_{i=1}^n$ be admissible trading speeds on zero-coupon bonds with maturities $\{T_i\}_{i=1}^n$ as defined in Section \ref{sec:main}, that is $v_{T_i} \in \mathcal A_{T_i}$ for any $i=1,...n$.
From \eqref{impacted_bond} and \eqref{impacted_coupon_bond_linear_combination} we get
\begin{subequations}
\begin{align*}
\tilde{B}(t,T) & = \sum_{i=1}^n c_i \tilde{P}(t,T_i) + N \tilde{P}(t,T_n) \\
& = B(t,T) - \sum_{i=1}^n c_i l(t,T_i) v_{T_i}(t) - N l(t,T_n) v_{T_n}(t)  \\
& \quad - \sum_{i=1}^n c_i K(t,T_i) y e^{-\rho t} - N K(t,T_n) ye^{-\rho t} \\
& \quad - \gamma \int_0^t e^{-\rho (t-s)} \left(\sum_{i=1}^n c_i K(t,T_i) v_{T_i}(s) + N K(t,T_n) v_{T_n}(s) \right) ds.
\end{align*}
\end{subequations}
Let us now assume that $l=\kappa K$ at all times and for all maturities, where $\kappa>0$ is a constant. Then, the impacted coupon bond dynamics can be written as
\begin{align}
\begin{split}
\tilde{B}(t,T) & = B(t,T) - y e^{-\rho t} \left[\sum_{i=1}^n c_i K(t,T_i) + N K(t,T_n) \right] \\
& \quad - \int_0^t e^{-\rho (t-s)} \kappa\delta(s-t) \left[\sum_{i=1}^n c_i K(t,T_i) v_{T_i}(s) + N K(t,T_n) v_{T_n}(s) \right] ds \\
& \quad - \gamma \int_0^t e^{-\rho (t-s)} \left[\sum_{i=1}^n c_i K(t,T_i) v_{T_i}(s) + N K(t,T_n) v_{T_n}(s) \right] ds,
\end{split}
\end{align}
where $\delta$ denotes the Dirac delta. Notice that under this assumption the impacted zero-coupon bond dynamics defined in \eqref{impacted_bond} boils down to
\begin{equation}\label{impacted_bond_simplified}
\tilde{P}(t,T) = P(t,T) - K(t,T) \left[y e^{-\rho t} + \int_0^t e^{-\rho (t-s)} v_T(s) \left(\gamma + \kappa\delta(s-t)\right) ds \right].
\end{equation}
This suggest we can define
\begin{equation*}
K^B(t,T) := \sum_{i=1}^n c_i K(t,T_i) + N K(t,T_n).
\end{equation*}
and the trading speed relative to the coupon bond as
\begin{equation}
v^B(t,s) := \frac{1}{K^B(t,T)} \left[\sum_{i=1}^n c_i K(t,T_i) v_{T_i}(s) + N K(t,T_n) v_{T_n}(s) \right].
\label{trading_speed_coupon_bond}
\end{equation}
Therefore, we obtain the following price impact model for the coupon bond: 
\begin{equation}\label{impacted_coupon_bond_simplified}
\tilde{B}(t,T) = B(t,T) - K^B(t,T) \left[y e^{-\rho t} + \int_0^t e^{-\rho (t-s)} v^B(t,s) \left(\gamma +  \kappa \delta(s-t)\right) ds \right].
\end{equation}
Interestingly, under the simplifying assumption that the functions $l$ and $K$ are equal up to some constant, we observe that the impacted zero-coupon bond $\tilde{P}(t,T)$ in \eqref{impacted_bond_simplified} and the impacted coupon bond $\tilde{B}(t,T)$ in \eqref{impacted_coupon_bond_simplified} are described by the same kind of dynamics.

This is particularly useful because, provided enough data on traded coupon bonds are available, one might attempt to use \eqref{trading_speed_coupon_bond} and \eqref{impacted_coupon_bond_simplified} to bootstrap the trading speeds $v_{T_i}$ relative to the zero-coupon bonds. Using the price impact model \eqref{impacted_bond}, it would be then possible to price impacted zero-coupon bonds consistently with market data. Finally, using these impacted zero-coupon bonds as building blocks, it would be possible to price, consistently with market data, more complicated and less liquid impacted interest rate derivatives, as discussed in Section \ref{subsec:pricing}.  

%----------------------------------------------------
%				HJM framework
%----------------------------------------------------
\subsection{HJM framework}\label{subsec:HJM}
In this section we turn our discussion to incorporating price impact into the Heath, Jarrow and Morton framework \cite{heath1992bond}, in order to model the forward curve. Notice that this approach, although it may look different, has some common aspects to the framework developed in Section \ref{subsec:def}. Namely, we start by adding artificially a price impact term to the forward rate dynamics. This corresponds to adding price impact to zero-coupon bonds in Section \ref{subsec:def}. The important difference is that, here, we are creating an impacted interest rate, which was not done in Section \ref{subsec:def}. Then we will develop the connection between the price impact of zero-coupon bonds and the price impact term incorporated into the forward rate, in order to reveal the financial interpretation of the latter. Note that both the zero-coupon bonds and the forward rate can be used as building blocks for the whole interest rates theory. We are therefore interested in showing the connection between the two in the presence of price impact. For a thorough discussion on the HJM framework in the classic interest rate theory, we refer to Chapter 6 of the book by Filipovic \cite{filipovic2009term}. 

Given an integrable initial forward curve $T \to \tilde{f}(0,T)$, we assume that the impacted forward rate process $\tilde{f}(\cdot,T)$ is given by 
\begin{equation}
\tilde{f}(t,T) = \tilde{f}(0,T) + \int_0^t\left( \alpha(s,T) + J^f(s,T)\right) ds + \int_0^t \sigma(s,T) d W^{\mathbb{P}}(s),
\label{impacted_forward_rate}
\end{equation}
for any $0\leq t \leq T$ and each maturity $T>0$. Here $W^{\mathbb{P}}$ is a Brownian motion under the measure $\mathbb{P}$ and $\alpha(\cdot,T)$, $J^f(\cdot,T)$ and $\sigma(\cdot,T)$ are assumed to be progressively measurable processes and satisfy for  any $T>0$
\begin{align*}
\begin{split}
\int_0^T \int_0^T (|\alpha(s,t)| + |J^f(s,t)| )ds dt & < \infty, \\
\sup_{s,t \leq T} |\sigma(s,t) | & < \infty.
\end{split}
\end{align*}
While the roles of $\alpha(\cdot,T)$ and $\sigma(\cdot,T)$ above are as in standard HJM model, the stochastic process $J^f$ represents the impact density relative to the forward rate and accounts for the fact that the forward curve is affected by the trading activity. From a modelling perspective, it plays a completely analogous role as the quantity $J_T$ defined in \eqref{impact_density} for the impacted zero-coupon bond. In fact, in Proposition \eqref{prop:relationship_Jf_Jp} we will determine the mathematical relationship linking these two quantities. Such relationship will allow us to understand how the forward curve is impacted by trading zero-coupon bonds. 

In this framework, the impacted short rate model is given by
\begin{equation}
\tilde{r}(t) := \tilde{f}(t,t) = \tilde{f}(0,t) + \int_0^t\left( \alpha(s,t) + J^f(s,t)\right)  ds + \int_0^t \sigma(s,t) d W^{\mathbb{P}}(s),
\label{impacted_short_rate}
\end{equation}
and the impacted zero-coupon bond is defined as follows 
\begin{equation}
\tilde{P}(t,T) = e^{-\int_t^T \tilde{f}(t,u) du}.
\label{impacted_bond_HJM}
\end{equation}

Next we derive the explicit dynamics of $\{\tilde{P}(t,T)\}_{0\leq t\leq T}$. The following corollary is an impacted version of Lemma 6.1 in \cite{filipovic2009term}.

%----------- Corollary: Impacted zero-coupon bond in HJM framework ---------%
\begin{corollary}[Impacted zero-coupon bond in HJM framework]\label{cor:impacted_zcb_HJM} 
For every maturity $T$ the impacted zero-coupon bond defined in \eqref{impacted_bond_HJM} follows the dynamics
\begin{equation}
\tilde{P}(t,T) = \tilde{P}(0,T) + \int_0^t \tilde{P}(s,T) \left(\tilde{r}(s) + \tilde{b}(s,T) \right) ds + \int_0^t \tilde{P}(s,T) \nu(s,T) d W^{\mathbb{P}}(s),\ \ \ t \leq T,
\label{impaced_ZCB_HJM}
\end{equation}
where $\tilde{r}$ is the impacted short rate defined in \eqref{impacted_short_rate} and 
\begin{equation} \label{b-v-rel}
\begin{aligned}
\nu(s,T) & := - \int_s^T \sigma(s,u) du, \\
\tilde{b}(s,T) & := - \int_s^T \alpha(s,u) du - \int_s^T J^f(s,u) du + \frac{1}{2} \nu^2(s,T).
\end{aligned}
\end{equation} 
\end{corollary}
%-----------------------------------------------------------------------------%

We now show that the impact $J^f$ can be expressed in terms of the impact relative to the zero-coupon bond, and vice versa. In order to show this correspondence in terms of agent's trading speed, we need to make an additional assumption on the trading speeds on zero-coupon bonds. We assume that the price impact in the forward curve is a result of trading by one or many agents over a continuum of zero-coupon bonds with maturities $T$ and trading speeds $\{T \geq 0 \, : \,  v_T \in \mathcal A_T\}$ so that 
\begin{equation} \label{sp-der}  
|\partial_T v_T(t)| <\infty, \quad \textrm{for all }  0\leq t \leq T,  \quad \mathbb{P}-\textrm{a.s.} 
\end{equation} 
Note that this assumption in fact makes sense in bond trading, which has discrete maturities, as it claims that when there is a highly traded $T_i$-bond, you would find that also the neighbouring $T_{i-1}$, $T_{i+1}$ are liquid. Assumption \eqref{sp-der} implies that $\partial_T I_T(t)$ is well defined as needed in the following Proposition.  We recall that $f$ represents the unimpacted forward rate which is given by setting $J^{f} \equiv 0$ in \eqref{impacted_forward_rate}.
 
%------------------------------ Proposition ----------------------------------%
%		 Relationship between forward rate impact and zero-coupon bond impact 
\begin{proposition}[Forward rate and zero-coupon bond price impact relation] \label{prop:relationship_Jf_Jp}
Let $I_T(t)$ be the overall impact defined in \eqref{def:overall_price_impact} and $\tilde{P}(\cdot,T)$ the impacted zero-coupon bond price in \eqref{impacted_bond_HJM}. Assume $\tilde{f}(0,t) =  f(0,t)$, meaning that the initial value of the forward curve is not affected by trading. Then, the forward rate impact $J^f$ introduced in \eqref{impacted_forward_rate} is given by
\begin{equation}\label{relationship_Jf_Jp}
J^f(t,T) = - \frac{\partial}{\partial T} \log \left( 1- \frac{I_T(t)}{P(t,T)} \right) , \quad \textrm{for all } 0\leq t \leq T \ \textrm{such that }\tilde{P}(t,T) >0.  
\end{equation}
\end{proposition}

The proof of Proposition \ref{prop:relationship_Jf_Jp} is given in Section \ref{sec:proofs}.

\begin{remark} Note that the requirement that $\tilde{P}(t,T)>0$ ensures that the logarithm on the right-hand side of \eqref{relationship_Jf_Jp} is well defined, as \eqref{bla} in the proof suggests. The proof also gives another relation between $J^{f}(\cdot, T)$ and $I_{T}$ which always holds but is perhaps not as direct. 
\end{remark} 
%-------------------------------------------------------------------------------%

A well known feature of the classic HJM framework is that, under the risk neutral measure, the drift of the forward rate is completely specified by the volatility through the so called \emph{HJM condition}. In order to understand how this condition is affected by the introduction of price impact, we will follow Theorem 6.1 of \cite{filipovic2009term}. In particular, we have the following key result.

%-------------------- Theorem: HJM condition with price impact ---------------%
\begin{theorem}[HJM condition with price impact]\label{hjm_condition_market_impact}
Let $\mathbb{P}$ be the real world measure under which the impacted forward rate as in \eqref{impacted_forward_rate}. Let $\tilde{\mathbb{Q}} \sim \mathbb{P}$ be an equivalent probability measure of the form
\begin{equation}\label{def:Q_tilde_rnd_HJM}
\frac{d \tilde{\mathbb{Q}}}{d \mathbb{P}} = \exp \left\{\int_0^t \tilde{\gamma}(s) d W^{\mathbb{P}}(s) - \frac{1}{2} \int_0^t \tilde{\gamma}^2(s) ds \right\},
\end{equation}
for some progressively measurable stochastic process $\tilde{\gamma}= \{\tilde \gamma(t) \}_{t\geq 0}$ such that $\int_0^t \tilde{\gamma}^2(s) ds < \infty$, for all $t>0$, $\mathbb P$-a.s. Then, $\tilde{\mathbb{Q}}$ is an equivalent (local) martingale measure if and only if 
\begin{equation} 
\tilde{b}(t,T) = - \nu(t,T) \tilde{\gamma}(t), \quad \textrm{for all } 0\leq t \leq T. 
\label{HJM_condition}
\end{equation}
with $\tilde{b}(\cdot,T)$ and $\nu(\cdot,T)$ defined as in \eqref{b-v-rel}. In this case, the dynamics of the impacted forward rate under the measure $\tilde{\mathbb{Q}}$ is given by 
\begin{equation}
\tilde{f}(t,T) = \tilde{f}(0,T) + \int_0^t \left( \sigma(s,T) \int_s^T \sigma(s,u) du \right) ds + \int_0^t \sigma(s,T) d W^{\tilde{\mathbb{Q}}}(s).
\label{impacted_forward_rate_Q_tilde}
\end{equation}
Moreover, the prices of impacted zero-coupon bonds are
\begin{equation}
\tilde{P}(t,T) = \tilde{P}(0,T) + \int_0^t \tilde{P}(s,T) \tilde{r}(s) ds + \int_0^t \tilde{P}(s,T) \nu(s,T) d W^{\tilde{\mathbb{Q}}}(s).
\label{impacted_zc_bond_Q_tilde}
\end{equation}
\end{theorem}

The proof of Theorem \ref{hjm_condition_market_impact} is given in Section \ref{sec:proofs}.
%------------------------------------------------------------------------------%

In our context such a measure $\tilde{\mathbb{Q}}$ would be clearly interpreted as an impacted risk-neutral measure, completely analogous to the measure defined in \eqref{radon_nikodym_der_Q_tilde}. In fact, the stochastic process $\tilde{\gamma}$ in the HJM condition \eqref{HJM_condition} is the counterpart in the HJM framework, of the impacted market price of risk $\tilde{\lambda}$ defined in Section \ref{sec:main}. Indeed, combining equations \eqref{market_price_of_risk} and \eqref{impaced_ZCB_HJM} we obtain for $0\leq t\leq T$,
\begin{equation*}
\tilde{\lambda}(t) = \frac{\tilde{P}(t,T) \left(r(t) - \tilde{b}(t,T) \right) - r(t) \tilde{P}(t,T)}{\tilde{P}(t,T) \nu(t,T)} = - \frac{-\tilde{b}(t,T)}{\nu(t,T)} = \tilde{\gamma}(t).
\end{equation*}

The HJM framework adjusted with price impact discussed in this section is therefore perfectly consistent with the price impact model for zero-coupon bonds introduced in Section \ref{subsec:def}.

We remark once again that, in the classic theory of interest rates, the meaning of the HJM condition lies in the fact that the drift of the forward rate is constrained under the risk neutral measure. Similarly, looking at the market price of risk, we notice that a constraint is present for the drift of the zero-coupon bond. In particular, its drift, under the risk neutral measure, has to be precisely the risk free interest rate. The interesting point we would like to make here is that, once we incorporate price impact, the same kind of constraints holds, only under the newly defined impacted measure $\tilde{\mathbb{Q}}$.

We conclude this section by making two remarks. We first address the question of when the measure defined in Theorem \ref{hjm_condition_market_impact} is an equivalent martingale measure, instead of just local martingale measure. The second remark concerns the Markov property of the impacted short rate. In both cases, we see that the classic results carry over to the price impact framework, thanks to the key fact that the impact component affects only the drift of the forward rate.

%-------------- Remark: Impacted risk neutral measure is an EMM ---------%
\begin{remark}[Impacted risk neutral measure is an EMM]
Let $\nu(t,T)$ be defined as in \eqref{b-v-rel}. From Corollary 6.2 of \cite{filipovic2009term} it follows that the measure $\tilde{\mathbb{Q}}$ defined in Theorem \ref{hjm_condition_market_impact} is an equivalent martingale measure if either
\begin{equation*}
\mathbb{E}^{\tilde{\mathbb{Q}}} \left[ e^{\frac{1}{2} \int_0^T \nu^2(t,T) dt }\right] < \infty, \quad \text{for all} \ T \geq 0, 
\end{equation*}
or 
\begin{equation*}
f(t,T) \geq 0, \quad \textrm{for all } 0\leq t\leq T. 
\end{equation*}
\end{remark}

%---------------- Remark: Markovianity short rate -----------------------%
\begin{remark}[Markov property of the short rate]
As pointed out in Chapter 5 of \cite{brigo2007interest}, one of the main drawbacks of HJM theory is that the implied short rate dynamics is usually not Markovian. Here we simply remark that, since the volatility of the forward rate $\sigma(t,T)$ is not affected by price impact, if the Markov property of the short rate $r(t)$ is ensured under the measure $\mathbb{Q}$ when there is no trading, hence no price impact, then it is also preserved in the presence of price impact under the $\tilde{\mathbb{Q}}$.
\end{remark}

%%%%%%%%%%%%%%%%%%%%%%%%%%%%%%%%%%%%%%%%%%%%%%%%%%%%%
%				Examples
%%%%%%%%%%%%%%%%%%%%%%%%%%%%%%%%%%%%%%%%%%%%%%%%%%%%%
\section{Examples}{\label{sec:examples}}	

%-----------------------------------------------------------------------
%		Pricing impacted Eurodollar futures with Vasicek model
%-----------------------------------------------------------------------
\subsection{Pricing impacted Eurodollar futures with Vasicek model}{\label{subsec:vasicek_example}}
In this section we illustrate the argument outlined in Section \ref{subsec:pricing} by computing the explicit price of a Eurodollar-futures contract when the underlying short rate follows an Ornstein-Uhlenbeck process \cite{vasicek1977equilibrium}. The dynamics under the risk neutral measure $\mathbb{Q}$ is given by
\begin{equation}
d r(t) = k (\theta - r(t)) dt + \sigma d W^\mathbb{Q}(t),
\label{sde:vasicek}
\end{equation}
with $k,\theta,\sigma$ positive parameters. The dynamics of the short rate under the real world measure $\mathbb{P}$ can be expressed as
\begin{equation}
d r(t) = k (\theta - r(t)) dt + \sigma (d W^\mathbb{P}(t) - \lambda(t) dt),
\label{short_rate_P_lambda}
\end{equation}
where we highlight the classic market price of risk process $\lambda$ defined in \eqref{market_price_of_risk}. Another representation for $r(t)$ under $\mathbb{P}$ is 
\begin{equation}
d r(t) = \tilde{k} (\tilde{\theta} - r(t)) dt + \sigma (d W^\mathbb{P}(t) - \tilde{\lambda}(t) dt),
\label{short_rate_P_lambda_tilde}
\end{equation}
where $\tilde{\lambda}$ is the impacted market price of risk defined in \eqref{def:lambda_tilde} and $\tilde k, \tilde \theta$ are positive constants. Combining the two equivalent representations \eqref{short_rate_P_lambda} and \eqref{short_rate_P_lambda_tilde}, we see that the following holds for any $t\geq 0$
\begin{equation}
k \theta - k r(t) - \sigma \lambda(t) = \tilde{k} \tilde{\theta} - \tilde{k} r(t) - \sigma \tilde{\lambda}(t).
\label{relationship_vasicek_pars}
\end{equation}
Similarly to what is done in the standard theory (see Brigo and Mercurio \cite{brigo2007interest}, section 3.2.1), we assume the short rate $r(t)$ has the same kind of dynamics under the measures $\mathbb{P}$, $\mathbb{Q}$ and $\tilde{\mathbb{Q}}$, that is 
\begin{equation}
\lambda(t) = \lambda r(t), \ \ \ \tilde{\lambda}(t) = \tilde{\lambda} r (t),
\label{lambda_lambda_t_vasicek}
\end{equation}
with $\lambda,\tilde{\lambda}$ constants. The whole impact is then encapsulated in the constant $\tilde{\lambda}$. By plugging \eqref{lambda_lambda_t_vasicek} into \eqref{relationship_vasicek_pars}, we deduce
\begin{align}\label{tilde_pars}
\begin{split}
\tilde{k} & = k - \sigma (\tilde{\lambda} - \lambda), \\ 
\tilde{\theta} & = \frac{k \theta}{k - \sigma (\tilde{\lambda} - \lambda)}. 
\end{split}
\end{align}
Clearly, in order to ensure all parameters are positive, we must require
\begin{equation*}
k > \sigma (\tilde{\lambda} - \lambda).
\end{equation*}

In this way, the short rate $r(t)$ is normally distributed under all three measures. In particular, plugging the Girsanov transformation from the measure $\mathbb{P}$ to the measure $\mathbb{\tilde{Q}}$, defined in \eqref{radon_nikodym_der_Q_tilde}, into equation \eqref{short_rate_P_lambda_tilde}, the short rate dynamics under $\tilde{\mathbb{Q}}$ can be conveniently rewritten as
\begin{equation} \label{vasicek_under_Q_tilde}
d r(t) = \tilde{k}(\tilde{\theta} - r(t)) dt + \sigma d W^{\tilde{\mathbb{Q}}}(t).
\end{equation}

Since the short rate under $\tilde{\mathbb{Q}}$ is Gaussian, $\{\int_t^T r(s) ds\}_{t\geq 0}$ is also a Gaussian process. At the same time, we recall the well known fact that if $X$ is a normal random variable with mean $\mu_X$ and variance $\sigma^2_X$, then $\mathbb{E}(\exp(X)) = \exp (\mu_X + \frac{1}{2} \sigma_X^2)$. Following the same argument as in (Brigo and Mercurio \cite{brigo2007interest}, Chapters 3.2.1, 3.3.2 and Chapter 4), we can use \eqref{vasicek_under_Q_tilde} in order to  express the impacted zero-coupon bond price as follows
% Appendix A, Hull White: Page 217 (PDF)
\begin{equation*}
\tilde{P}(t,T) = A(t,T) e^{-B(t,T) r(t)}
\end{equation*}
where
\begin{align}\label{A_B_coeff_vasicek}
\begin{split}
A(t,T) & = \exp \left\{ \left(\tilde{\theta} - \frac{\sigma^2}{2 \tilde{k}^2} \right) [B(t,T) - T + t] - \frac{\sigma^2}{4 \tilde{k}} B^2(t,T) \right\}, \\
B(t,T) & = \frac{1}{\tilde{k}} \left(1 - e^{-\tilde{k} (T-t)} \right). 
\end{split}
\end{align}
Hence, the key expectation needed to compute the impacted Eurodollar future fair price in equation \eqref{fair_price_impacted_Eurodollar} is equal to
\begin{equation} \label{eu-d}
\mathbb{E}^{\tilde{\mathbb{Q}}} \left[ \frac{1}{\tilde{P}(t,T)} \right] = \frac{1}{A(t,T)} \mathbb{E}^{\tilde{\mathbb{Q}}} \left[e^{B(t,T) r(t)} \right].
\end{equation}
Since $r(t)$ is normally distributed, $B(t,T) r(t)$ will be normally distributed as well with mean and variance respectively equal to (see Brigo and Mercurio \cite{brigo2007interest}, Eq. (3.7))
\begin{subequations}
\begin{align*}
\mathbb{E}^{\tilde{\mathbb{Q}}}[B(t,T) r(t)] & = B(t,T) \left[r(0) e^{-\tilde{k}t} + \theta (1-e^{-\tilde{k}t}) \right], \\
\text{Var}^{\tilde{\mathbb{Q}}}[B(t,T) r(t)] & = B^2(t,T) \left[ \frac{\sigma^2}{2 \tilde{k}} (1 - e^{-2 \tilde{k} t}) \right].
\end{align*}
\end{subequations}
Therefore in order to get the impacted price of a Eurodollar-future  contract we need to compute the expectation in the right hand side of \eqref{eu-d} which can be written explicitly as
\begin{multline*}
\mathbb{E}^{\tilde{\mathbb{Q}}} \left[ \frac{1}{\tilde{P}(t,T)} \right] = \frac{1}{A(t,T)} \times \\
\times \exp \left\{ B(t,T) [r(0) e^{- \tilde{k} t} + \theta (1-e^{- \tilde{k}t})] + \frac{1}{2} B^2(t,T) \left[ \frac{\sigma^2}{2\tilde{k}} (1-e^{-2 \tilde{k} t}) \right] \right\}.
\end{multline*}
The main conclusion here is that defining the short rate under the impacted risk neutral measure preserves analytical tractability of interest rate derivatives precises. 

%-----------------------------------------------------------------------
%		Pricing impacted Eurodollar futures with Hull White model
%-----------------------------------------------------------------------
\subsection{Pricing impacted Eurodollar futures with Hull White model}{\label{subsec:hullwhite_example}}
In this section we compute the explicit price of a Eurodollar-futures contract when the underlying short rate follows a Hull White model \cite{hull1990pricing}. We start with the classic framework where there is not price impact. In this case the short rate is given by 
\begin{equation*}
d r(t) = \left[ \theta(t) - a r(t) \right] dt + \sigma d W^{\mathbb{Q}}(t),
\end{equation*}
where $a$ and $\sigma$ are positive constants and the function $\theta$ is chosen in order to fit exactly the term structure of interest rates being currently observed in the market. Denoting by $P^M(0,T)$ the unimpacted market discount factor for the maturity $T$ and defining the (unimpacted) market instantaneous forward rate at time $0$ for the maturity $T$
\begin{equation*}
f^M(0,T) := - \frac{\partial}{\partial T} \ln P^M(0,T),
\end{equation*}
the function $\theta$ is given by (see e.g. Brigo and Mercurio \cite{brigo2007interest}, Chapter 3, Eq. (3.34))
\begin{equation*}
\theta(t) = \frac{\partial f^M(0,t)}{\partial T} + a f^M(0,t) + \frac{\sigma^2}{2 a} \left(1-e^{-2 a t} \right),
\end{equation*}
where $\frac{\partial f^M(0,t)}{\partial T}$ denotes the partial derivative of $f^M$ with respect to its second variable. We start by computing the price under the classic risk neutral measure $\mathbb{Q}$. According to eq. (3.36)--(3.37) in Chapter 3 of \cite{brigo2007interest}, the short rate is normally distributed with mean and variance respectively equal to 
\begin{subequations}
\begin{align*}
\mathbb{E}^\mathbb{Q}[r(t) | \mathcal{F}_s] & = r(s) e^{-a (t-s)} + \alpha(t) - \alpha(s) e^{-a (t-s)} \\
\text{Var}^\mathbb{Q}[r(t) | \mathcal{F}_s] & = \frac{\sigma^2}{2 a} \left[1 - e^{-2 a (t-s)} \right],
\end{align*}
\end{subequations}
where 
\[
\alpha(t) := f^M(0,t) + \frac{\sigma^2}{2 a^2} (1-e^{-a t})^2.
\]
As before, we notice that the integral of the short rate will be normally distributed as well, hence the price of a zero-coupon bond under the classic risk neutral measure is given by (see eq. (3.39) in Chapter 3 of \cite{brigo2007interest}), 
\begin{equation*}
P(t,T) = A(t,T) e^{-B(t,T) r(t)},
\end{equation*}
where
\begin{align*}
A(t,T) &= \frac{P^M(0,T)}{P^M(0,t)} \exp \left\{B(t,T) f^M(0,t) - \frac{\sigma^2}{4 a} (1-e^{-2 a t}) B^2(t,T)\right\},\\
B(t,T) &= \frac{1}{a} \left[1-e^{-a (T-t)} \right].
\end{align*}
Moreover, the term $B(t,T) r(t)$ is still normally distributed and we immediately have
\begin{subequations}
\begin{align*}
\mathbb{E}^\mathbb{Q}[ B(t,T) r(t) | \mathcal{F}_s] & = B(t,T) \left(r(s) e^{-a (t-s)} + \alpha(t) - \alpha(s) e^{-a (t-s)} \right), \\
\text{Var}^\mathbb{Q}[ B(t,T) r(t) | \mathcal{F}_s] & = B^2(t,T) \frac{\sigma^2}{2 a} \left[1 - e^{-2 a (t-s)} \right].
\end{align*}
\end{subequations}
This implies that the expectation we are interested in, under the classic risk neutral measure $\mathbb{Q}$, can be written explicitly as (see Section 13.12.1 in \cite{brigo2007interest})
\begin{multline*}
\mathbb{E}^\mathbb{Q} \left[ \frac{1}{P(t,T)} \right] = \frac{1}{A(t,T)} \exp \left\{B(t,T) \mathbb{E}^\mathbb{Q}[r(t)] + \frac{1}{2} B^2(t,T) \text{Var}^\mathbb{Q}[r(t)] \right\}.
\end{multline*}
Next we derive the corresponding expression under the impacted risk neutral measure $\tilde{\mathbb{Q}}$ in \eqref{radon_nikodym_der_Q_tilde}. We assume as in Section \ref{subsec:vasicek_example} that the market price of risk and impacted market price of risk are given by
\[
\lambda(t) = \lambda r(t), \ \ \ \tilde{\lambda}(t) = \tilde{\lambda} r(t),
\]
for some constants $\lambda,\tilde{\lambda}$. Using the Girsanov change of measure from $\mathbb{Q}$ to $\tilde{\mathbb{Q}}$ defined in Section \ref{subsec:rnm}, it follows that the short rate under $\tilde{\mathbb{Q}}$ is given by  
\begin{subequations}
\begin{align*}
d r(t) & = \left[ \theta(t) - a r(t) \right] dt + \sigma d W^{\mathbb{Q}}(t) \\
& = \left[\theta(t)- a r(t) \right] dt + \sigma d W^{\tilde{\mathbb{Q}}}(t) + \sigma(\tilde{\lambda} - \lambda) r(t) dt \\
& = \left[ \theta(t) - (a - \sigma (\tilde{\lambda} - \lambda)) r(t) \right] dt + \sigma d W^{\tilde{\mathbb{Q}}}(t).
\end{align*}
\end{subequations}
Hence, we can define the impacted parameter
\begin{equation*}
\tilde{a} := a - \sigma (\tilde{\lambda} - \lambda).
\end{equation*}
The pricing formula for $\mathbb{E}^{\tilde {\mathbb{Q}}}\left[ \frac{1}{P(t,T)} \right]$ is then derived by following the same steps as in the classic case. Similarly to the Vasicek model, analytical tractability is preserved.

%%%%%%%%%%%%%%%%%%%%%%%%%%%%%%%%%%%%%%%%%%%%%%%%%%%%%
%				Numerical results
%%%%%%%%%%%%%%%%%%%%%%%%%%%%%%%%%%%%%%%%%%%%%%%%%%%%%
\section{Numerical results}{\label{sec:numerical_results}}
In this section we give a few numerical examples for the behaviour of the yield curve under price impact in the framework of short-rate affine models, which was described in Section \ref{subsec:yield}. In order to compute the cross price impact, we need the drift and the volatility of the zero-coupon bond. For the sake of simplicity, we assume the short rate is described by a Vasicek model \eqref{sde:vasicek}
\begin{equation*}
d r(t) = k (\theta - r(t)) dt + \sigma d W^\mathbb{Q}(t),
\end{equation*}
with $k,\theta,\sigma$ positive parameters. Then, the drift and the diffusion coefficients of the unimpacted zero-coupon bond are given by \eqref{affine:drift_vol_zero_coupon_bond}:
\begin{align*}
\mu_T(t,r(t)) &= e^{-B(t,T) r(t)} \bigg[\frac{\partial A}{\partial t} - A(t,T) \frac{\partial B}{\partial t} r(t) + \frac{1}{2} A(t,T) B^2(t,T) \sigma^2& \\
&\quad - A(t,T) B(t,T) k (\theta -r(t)) \bigg], \\
\sigma_T(t,r(t))& = - \sigma B(t,T) A(t,T) e^{-B(t,T) r(t)},
\end{align*}
where the functions $A,B$ are given as in \eqref{A_B_coeff_vasicek} and their derivatives are given by
\begin{equation*}
\frac{\partial B}{\partial t} = - e^{-k (T-t)}, \ \ \ \
\frac{\partial A}{\partial t} = A(t,T) \left[ \left(\theta - \frac{\sigma^2}{2 k^2} \right) \left(\frac{\partial B}{\partial t} + 1 \right) - \frac{\sigma^2}{2 k} B(t,T) \frac{\partial B}{\partial t} \right].
\end{equation*}
We can then plug all these quantities in equation \eqref{sde:cross_impacted_bond_S} to determine the dynamics of the cross-impacted zero-coupon bond and therefore the corresponding impacted yield. We set the following values for the parameters in \eqref{sde:vasicek}:
\begin{equation*}
k=0.20, \quad \theta = 0.10, \quad \sigma = 0.05, \quad r_0 = 0.01. 
\end{equation*}
We consider zero-coupon bonds with maturities $ \mathbb{T} := \left\{1,2,5,10,15 \right\}$ years and assume that an agent is trading on the bond with maturity $T=5$ years. All the other zero-coupon bonds experience cross price impact during the trading period. We fix the execution time horizon to be $\tau = 10$ days. All bonds are simulated over the time interval $[0, 9\ \text{months}]$, discretized in $N=365$ subintervals with time step $\Delta t = 1/365$. The short rate $r$ defined in \eqref{sde:vasicek} is simulated via Euler-Maruyama scheme. Since we are going to describe the average behaviour of the yield curve under market impact, we also set the number of Monte Carlo simulations to $M=10.000$. As we shall see below in the detailed algorithm, for each realization of the short rate, we will have a corresponding impacted yield curve. The idea is then to plot the average of such curves.

For the sake of simplicity, we discuss the benchmark trading strategy
\begin{equation}\label{def:benchmark_strategy}
v_T(s):=
\begin{cases} c, & \mbox{if } s \leq \tau \\ 0, & \mbox{otherwise }
\end{cases}
\end{equation}
with $c$ some positive constant if we buy, negative if we sell. In our simulations we choose $c=2$. The transient impact defined in \eqref{def:transient_impact} reads as
\begin{equation}
\Upsilon_T^v(t) = y e^{-\rho t} + \gamma e^{-\rho t} \int_0^t e^{\rho s} c \mathbbm{1}_{s \leq \tau} ds,
\label{benchmark_transient_impact}
\end{equation}
where the  parameters are set to
\begin{equation*}
\rho = 2, \ \ \ \gamma = 1, \ \ \ y = 0.01.
\end{equation*}
The functions $l,K$ introduced in \eqref{impacted_bond} are assumed to be of the form
\begin{equation*}
l(t,T) = \kappa \left(1-\frac{t}{T}\right)^\alpha, \ \ \ K(t,T) =  \left(1-\frac{t}{T}\right)^\beta
\end{equation*}
with $\kappa \geq 0, \alpha,\beta \geq 1$. In particular, we choose
\begin{equation*}
\alpha = 1, \ \ \beta=1, \ \ \kappa =0.01. 
\end{equation*}
Following \eqref{impacted_bond_differential_form}, the price of the impacted bond in $T=5$y is 
\begin{equation*}
\tilde{P}(t,T) = P(t,T) + \int_0^t J_T(s) ds,
\end{equation*}
where $J_T$, which was defined in \eqref{impact_density}, is specified to be
\begin{equation*}
J_T(t) = - \frac{\kappa}{T} v_T(t) + \left(1-\frac{t}{T} \right) \left[-\rho \Upsilon_T^v + v_T(t) \right] - \Upsilon_T^v(t)
\end{equation*}
The algorithm we implemented to simulate the impacted yield curve consists of the following steps.
\begin{steps}
\item Simulate a path of the short rate $r(t)$ given in \eqref{sde:vasicek} for $t \in [0,9\ \text{months}]$. 
\item Compute the unimpacted zero-coupon bond price $P(t,T)$ for the trading maturity  $T=5$ years using equation \eqref{affine:bond_price} for $t \in [0,9\ \text{months}]$.
\item Compute the unimpacted yield $Y(t,T)$ by plugging $P(t,T)$ in \eqref{def:classic_yield} for $t \in [0,9\ \text{months}]$.
\item Compute the (directly) impacted zero-coupon bond $\tilde{P}(t,T)$ using \eqref{impacted_bond_differential_form} for $t \in [0,9\ \text{months}]$.
\item Compute the (directly) impacted yield $\tilde{Y}(t,T)$ by plugging $\tilde{P}(t,T)$ into \eqref{def:impacted_yield} for $t \in [0,9\ \text{months}]$.
\item For all other maturities $S=1,2, 10, 15$ years, compute the cross impacted zero-coupon bond price $\tilde{P}(t,S)$ using equation \eqref{sde:cross_impacted_bond_S} for $t \in [0,9\ \text{months}]$.
\item Compute the cross impacted yield $\tilde{Y}(t,S)$ by plugging $\tilde{P}(t,S)$ into \eqref{def:impacted_yield} for $t \in [0,9\ \text{months}]$.
\item Repeat these steps $M=10.000$ times and compare the average of $Y(t,T)$ with the average of $\tilde{Y}(t,T)$.
\end{steps}

In Figure \ref{fig:imp_unimp_yield_average} we visualize for all maturities the average classic yield $\mathbb{E}[Y(t,T)]$ versus the average impacted yield $\mathbb{E}[\tilde{Y}(t,T)]$ at times $t=5$ days (middle of trading), $t=11$ days (right after trading is ended) and $t=270$ days (after $9$ months). 
\begin{figure}[H]
\centering
\includegraphics[scale=.50]{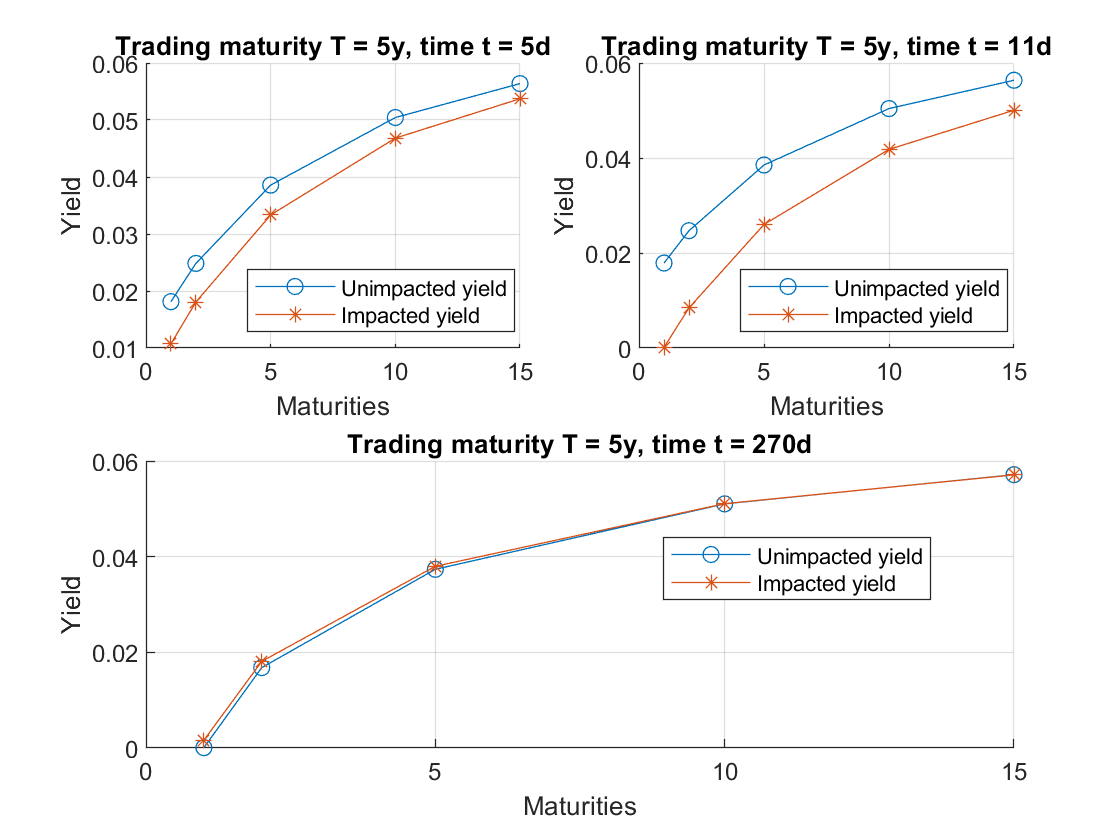}
\caption{Trading zero-coupon bond with maturity $T=5$ years. Average unimpacted yield curve and average impacted yield curve in the middle of trading (top left panel), right after trading is concluded (top right panel) and after nine months (bottom panel).}
\label{fig:imp_unimp_yield_average}
\end{figure}
In the top panels we see that the yield has decreased over all maturities as result of trading. This is consistent with the fact that a buy strategy of bonds pushes their prices up due to price impact, hence the yield decreases. Clearly, the almost parallel shift of the yield curve is a just a consequence of the very simple (constant) trading strategy we defined in equation \eqref{def:benchmark_strategy}. We expect to observe much more complicated behaviours when implementing more sophisticated strategies. In the bottom panel, instead, we observe that, roughly nine months after performing the trades, the two yield curves pretty much coincide. This is due to the transient component in the price impact model, which induces impacted yield curve to converge to its classic counterpart as time goes by. When analysing price impact due to zero-coupon bond trading, one aspect that certainly can't be ignored is the special nature of the assets we are trading. Unlike what happens with stocks, the time evolution of zero-coupon bonds is constrained, specifically by the fact that they must reach value $1$ at maturity. It therefore appears that two fundamental forces are in play: the intrinsic \emph{pull to par} effect, which makes both the impacted and unimpacted bond price go to $1$, hence the corresponding yields to $0$, and the \emph{price impact} effect, which induces the bond price to first increase (if we buy) or decrease (if we sell), and then revert back to its unimpacted value. Interestingly, when trading stocks, it will take the impacted asset forever to converge to the unimpacted counterpart, as the transient impact converges to $0$ as time $t$ goes to infinity. When trading bonds, though, this convergence occurs in finite time. In order to better understand the role played by price impact, in Figure \ref{fig:imp_classic_bonds} we compare directly the behaviour of the impacted bond $\tilde{P}(t,T)$ and of the classic bond $P(t,T)$ for different maturities $T$.
\begin{figure}[H]
\centering
\includegraphics[scale=.50]{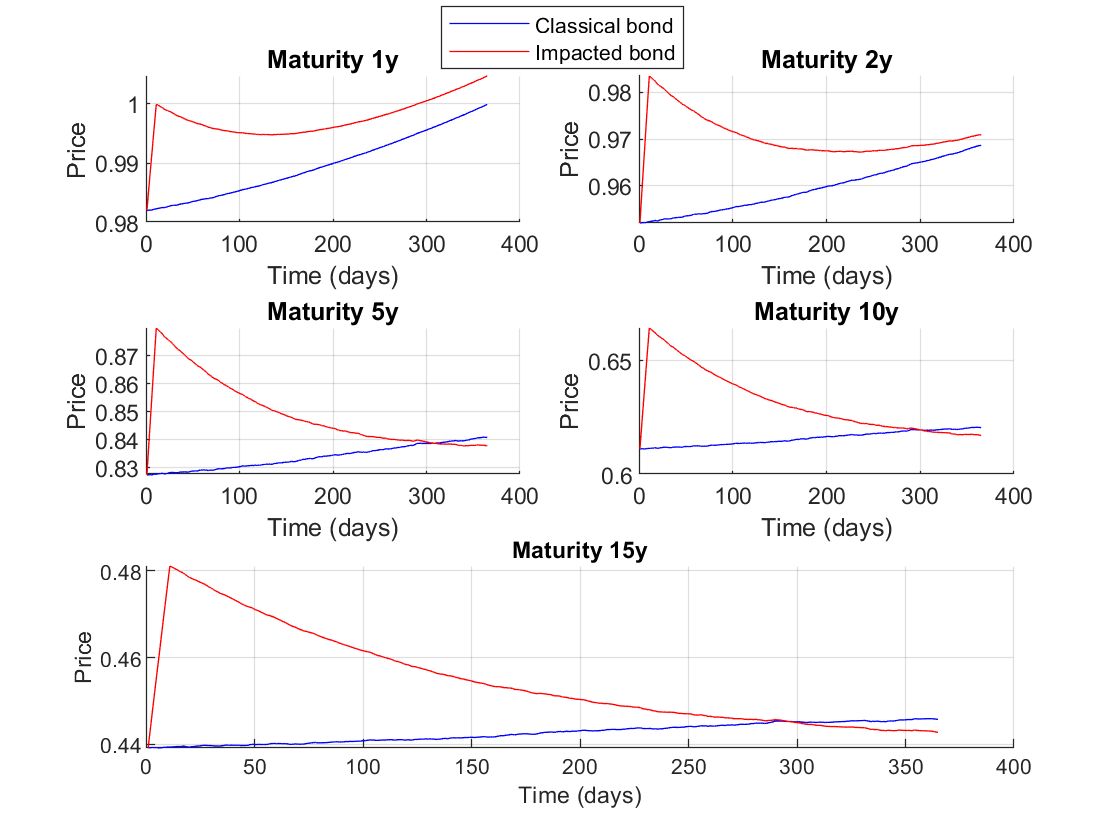}
\caption{Trading bond with maturity $T=5$ years. Averaged impacted zero-coupon bond vs averaged classic zero-coupon bond for maturities $1$ year (top left panel), $2$ years (top right panel), $5$ years (middle left panel), $10$ years (middle right panel), $15$ (bottom panel). All curves are observed over the interval $[0, 1\ \text{year}]$.}
\label{fig:imp_classic_bonds}
\end{figure}
We observe that, over one year, the pull-to-par effect is somehow stronger than the transient impact effect in bonds with short maturity $(T=1,2)$. By this, we mean that the unimpacted and impacted bonds meet at, or very close to, maturity. \footnote{It can be observed that the price of the cross-impacted zero-coupon bond with maturity $S=1$ year is not $1$ at expiration, as it should be, but slightly higher (top left panel). This is not a numerical error, but rather a consequence of our model not being able to ensure the cross-impacted bonds reach value precisely $1$ at their respective maturities. See Remark \ref{rem:cross_impacted_bonds_at_maturity}.} For bonds with long maturity ($T=5,10,15$), instead, the transient effect is prominent. This causes the impacted bond curve and the unimpacted bond curve to cross each other significantly before their maturity. In fact, we can numerically compute the first instant the two curves meet and we observe that the longer the maturity, the sooner this happens. This is illustrated in Figure \ref{fig:first_hitting_time}.
\begin{figure}[H]
\centering
\includegraphics[scale=.50]{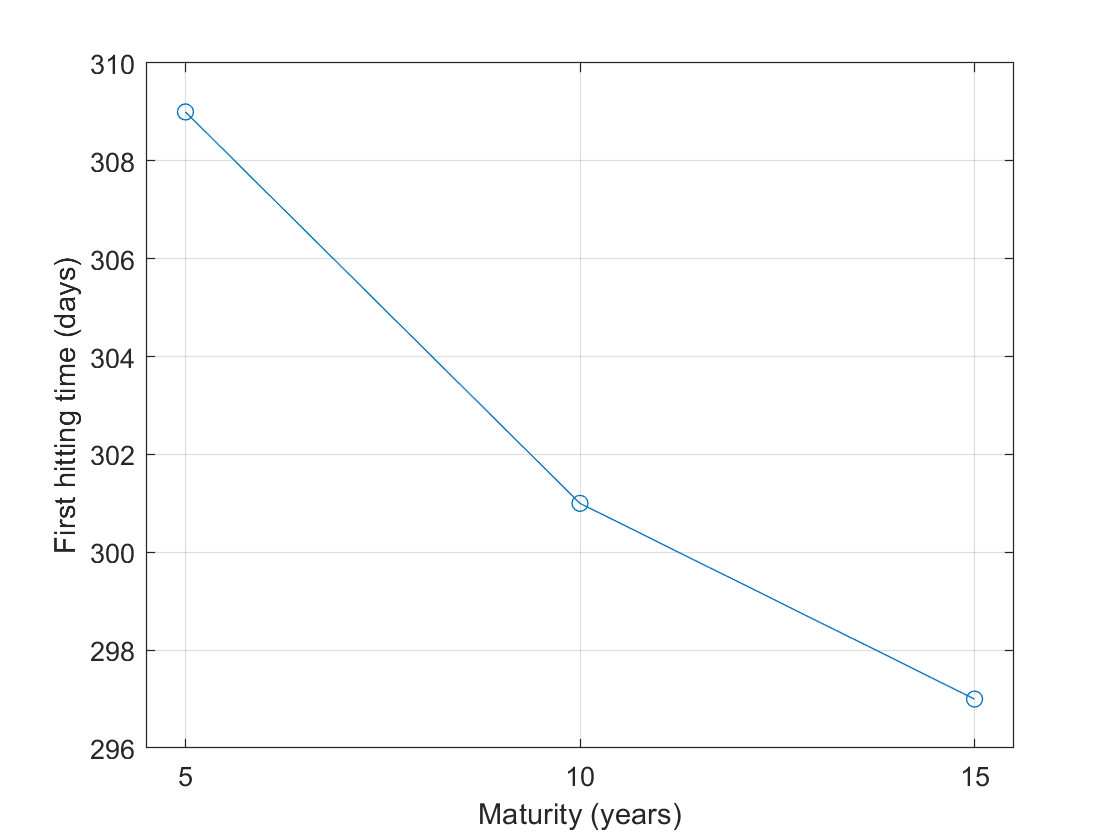}
\caption{Trading bond with maturity $T=5$ years. First instant (in days) impacted bond curve and unimpacted bond curve cross for maturities $5,10,15$ years. All curves are observed over the interval $[0, 1\ \text{year}]$.}
\label{fig:first_hitting_time}
\end{figure}

The interplay between the cross price impact effect, averaged over $10.000$ realizations, and the pull to par effect is demonstrated in Figure \ref{fig:average_imp_bond_1_rho} for the price of a zero-coupon bond with maturity $S=1$ year when trading a bond with maturity $T=5$ years. Trading takes place on the first 10 days of the year, while the time scale in the graph is of one year. We illustrate this effect for various values of the transient impact parameter $\rho$ in equation \eqref{benchmark_transient_impact}.
\begin{figure}[H]
\centering
\includegraphics[scale=.50]{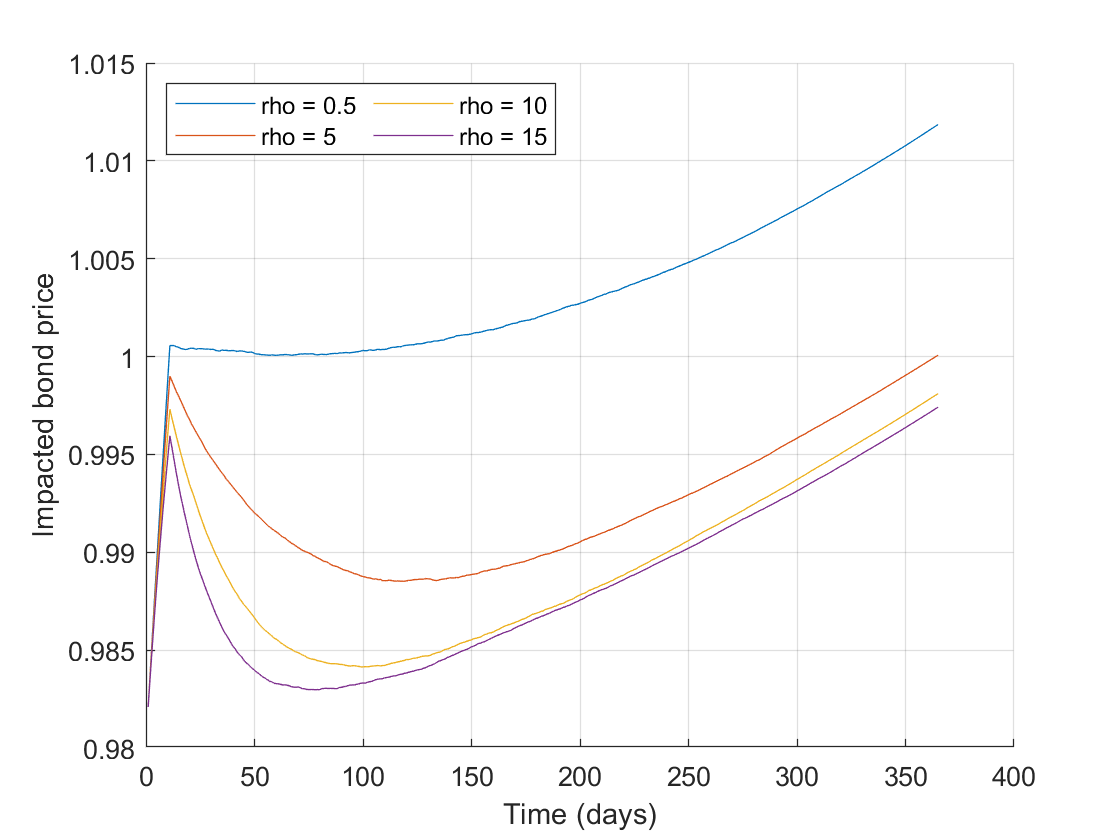}
\caption{Averaged cross price impact effect vs. pull to par effect over $10.000$ realizations is demonstrated for the price of a bond with maturity $S=1$ year when trading a bond with maturity $T=5$ years, for various values of $\rho$. Trading takes place on the first 10 days of the year, while the time scale in the graph is of one year.}
\label{fig:average_imp_bond_1_rho}
\end{figure}
It can be observed that the higher $\rho$, the more aggressively the price is "pulled down" close to the original price before the trades. At the beginning, far from maturity, the transient impact component dominates and the price decreases. After some time, though, the bond intrinsic nature takes over and the price starts to increase.

Another phenomenon which is revealed in our framework is the interplay between the mean reversion of the short rate model and the price impact. Recall that in Section \ref{subsec:vasicek_example} we found that the mean reversion speed $k$ under the measure $\mathbb{Q}$ and the mean reversion $\tilde{k}$ under the price-impacted measure $\tilde{\mathbb{Q}}$ are linked by \eqref{relationship_vasicek_pars} as follows
\[
\tilde{k} = k - \sigma(\tilde{\lambda}-\lambda),
\]
with $\tilde{\lambda}, \lambda$ representing the impacted market price of risk and the classic market price of risk respectively. We stress that the higher $k$, the faster the short rate $r$ under $\mathbb{Q}$ and its counterpart under $\tilde{\mathbb{Q}}$ converge to their respective stationary distributions. At the same time, since the variance of the stationary distribution is $\sigma^2/(2k)$, large values of $k$ reduce the overall variance of the model, thereby making the two types of rates that we consider closer to each other. This, in turn, implies that after a long time ($T=10,15$ years) the zero-coupon bond $P$ and impacted zero-coupon bond $\tilde{P}$, hence their yields, will be closer to each other. Conversely, if $k$ is small, the two short rates are quite far from each other and the overall variance of the model is large. Furthermore, looking again at \eqref{relationship_vasicek_pars}, we notice that the larger $k$, the less significant the impact component $-\sigma (\tilde{\lambda} - \lambda)$, and vice versa. In a way, the speed of mean reversion works in an opposite direction to the price impact. We demonstrate this in Figure \ref{fig:yield_curve_mean_reversion} for $k=0.01$ (top panel) and for $k=0.20$ (bottom panel). As above, we trade the zero-coupon bond with maturity $T=5$ years, trading occurs for the first $10$ days and the yields are observed after $9$ months. The difference in behaviour is evident for long maturities $(T=5,10,15)$. While in the bottom panel unimpacted yield and impacted yield are really close to each other (as in Figure \eqref{fig:imp_unimp_yield_average}, right panel), in the top panel the distance between the two yields is rather significant.
\begin{figure}[H]
\centering
\includegraphics[scale=.50]{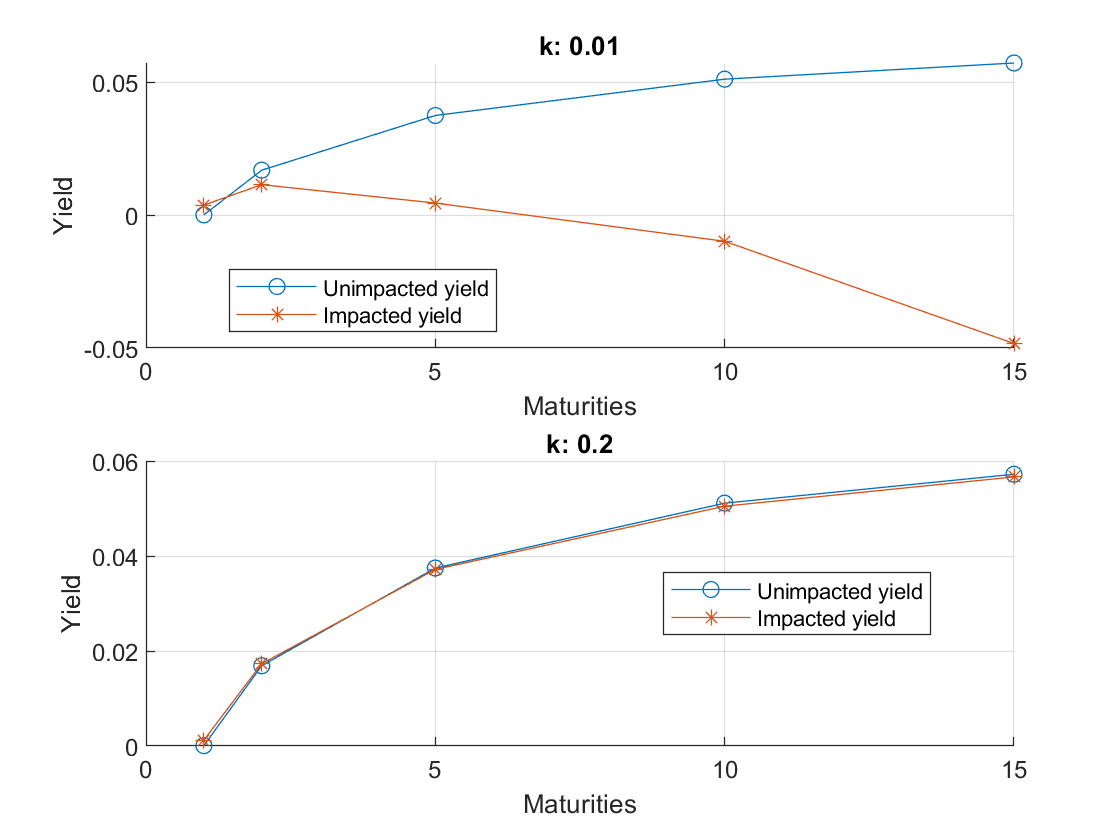}
\caption{Impacted and unimpacted yield curves for $k=0.01$ (top panel) and for $k=0.20$ (bottom panel) when trading zero-coupon bond with maturity $T=5$ years. Trading occurs during the first $10$ days. Yield curves are observed after nine months.}
\label{fig:yield_curve_mean_reversion}
\end{figure}

%%%%%%%%%%%%%%%%%%%%%%%%%%%%%%%%%%%%%%%%%%%%%%%%%%%%%%%%%%%
%		Optimal execution with impacted bonds
%%%%%%%%%%%%%%%%%%%%%%%%%%%%%%%%%%%%%%%%%%%%%%%%%%%%%%%%%%%
\section{Optimal execution of bonds in presence of price impact }\label{sec:optexec}
In this section we consider a problem of an agent who tries to liquidate a large inventory of $T$-bonds within a finite time horizon $[0,\tau]$ where $\tau < T$. We assume that the agent's transactions create both temporary and transient price impact and that the performance of the agent is measured by a revenue-cost functional that captures the transaction costs which result by price impact, and the risk of holding inventory for long time periods.  Our optimal execution framework is closely related to the framework which was proposed for execution of equities in Section \ref{subsec:def} of \cite{neuman2020optimal}. The main difference between the two frameworks is that in our framework the price impact has to vanish at the bond's maturity in order to satisfy the boundary condition $\tilde P(T,T)=1$. 

Let $T>0$ denote the bond's maturity. We assume that the unimpacted bond price $P(\cdot,T)$ is given by \eqref{class-b} and we consider the canonical decomposition $P(\cdot,T)= A(\cdot,T) + \bar{M}(\cdot,T)$, where 
\begin{equation*}
A(t,T) := \int_0^t \mu_T(s,r(s)) ds, \quad 0\leq t\leq T, 
\end{equation*}
is a predictable finite-variation process and  
\begin{equation*}
\bar{M}(t,T) := \int_0^t \sigma_T(s,r(s)) d W^{\mathbb{P}}(s) , \quad 0\leq t\leq T, 
\end{equation*}
local martingale. We assume that the coefficients $\sigma_{T},\mu_{T}$ in \eqref{class-b} are such that we have 
\begin{equation} \label{def: mathcal_h_squared} 
\mathbb{E}[\langle \bar{M}(\cdot,T) \rangle_\tau] + \mathbb{E} \left[ \left( \int_0^\tau |d A (\cdot,T)| \right)^2 \right] < \infty.
\end{equation} 
In this case we say that a bond price $\{P(t,T)\}_{t \in [0,T]}$ is in $\mathcal{H}^2$.

The initial position of the agent's inventory is denoted by $x>0$ and the number of shares the agent holds at time $t\in [0,\tau]$ is given by 
\begin{equation} \label{def:X}
X^{v_T}(t)\triangleq x-\int_0^t v_T(s)ds
\end{equation}
where $\{v_T(t)\}_{t \in [0,\tau]}$ denotes the trading speed. We say that the trading speed is admissible if it belongs to the following class of admissible strategies 
\begin{equation} \label{def:admissset} 
\mathcal A \triangleq \left\{ v_T \, : \, v_T \textrm{ progressively measurable s.t. } \mathbb E\left[ \int_0^\tau v^2_T(s) ds \right] <\infty \right\}.
\end{equation} 
We assume that the trader's trading activity causes price impact on the bond's price as described by $\{\tilde P(t,T)\}_{t \in [0,T]}$ in \eqref{impacted_bond}. 

As in Section 2 of \cite{neuman2020optimal}, we now suppose that the trader's optimal trading objective is to unwind her initial position $x>0$ in the presence of temporary and transient price impact through maximizing the following performance functional 
\begin{equation}  \label{j-fun}
\begin{aligned}
\mathcal{J}(v) := \mathbb{E} \bigg [\int_0^\tau \bigg(P(t,T) - K(t,T) \Upsilon_T^v(t) \bigg) v_T(t) dt - \int_0^\tau l(t,T) v_T^2(t) dt + X_T^v(\tau) P(\tau,T)  \\
- \phi \int_0^\tau (X_T^v(t))^2 dt - \varrho (X_T^v(\tau))^2 \bigg].
\end{aligned}
\end{equation} 
The first, second and third terms in $\mathcal J$ represent the trader's terminal wealth, meaning the final cash position including the accrued trading costs induced by temporary and transient price impact, as well as the remaining final risky asset position's book value. The fourth and fifth terms, instead, account for the penalties $\phi,\varrho>0$ on the trader's running penalty (i.e. the risk aversion term) and the penalty of holding any terminal inventory, respectively. 

Since $T$ is fixed, for the sake of readability we will omit the subscripts $T$ for the rest of this section. 
The main result of this section is the derivation of the unique optimal admissible strategy, namely 
\begin{equation}
\mathcal{J}(v) \to \max_{v \in \mathcal{A}}.
\label{opt_stoch_control}
\end{equation}
and exhibiting an explicit expression for the optimal trading strategy. We define
\begin{equation}
A(t) :=
\left(
\begin{matrix}
0 & 0 & -1 & 0 \\
0 & -\rho & \gamma & 0 \\
-2 \phi \Lambda(t) & \begin{matrix} \rho K(t,T) \Lambda(t)\\ \hfill{} - \Lambda'(t) K(t,T) - \Lambda(t) \partial_t K(t,T) \end{matrix}  & 0 & \Lambda'(t) + \rho \Lambda(t) \\
0 & 0 & K(t,T) \gamma & \rho
\end{matrix}
\right),
\label{def:time_dep_matrix_A}
\end{equation}
where
\begin{equation}
\Lambda(t) := \frac{1}{2 l(t,T)}. 
\label{def:Lambda_optimal_exe}
\end{equation}
Note that $\Lambda(t)$ is well defined for $0\leq t \leq \tau$ since $l(t,T)>0$ on this interval by \eqref{l-pos}. Let $\Phi$ be the fundamental solution of the matrix-valued ordinary differential equation
\begin{align}\label{matrix_ODE}
\begin{split}
\frac{d}{dt} \Phi(t) & = A(t) \Phi(t), \\
\Phi(0) & = \text{Id}.
\end{split}
\end{align}
Let us define the matrix
\begin{equation}\label{def:psi_matrix}
\Psi(t,\tau) := \Phi^{-1}(\tau) \Phi(t). 
\end{equation}
We also define the vector $G$:
\begin{align}\label{def:G_components}
\begin{split}
G^1(t,\tau) & := \frac{\varrho}{l(\tau,T)} \Psi^{11}(t,\tau) - \frac{K(\tau,T)}{2 l(\tau,T)} \Psi^{21}(t,\tau) - \Psi^{31}(t,\tau) ,\\
G^2(t,\tau) & := \frac{\varrho}{l(\tau,T)} \Psi^{12}(t,\tau) - \frac{K(\tau,T)}{2 l(\tau,T)} \Psi^{22}(t,\tau) - \Psi^{32}(t,\tau)  ,\\
G^3(t,\tau) & := \frac{\varrho}{l(\tau,T)} \Psi^{13}(t,\tau) - \frac{K(\tau,T)}{2 l(\tau,T)}\Psi^{23}(t,\tau) - \Psi^{33}(t,\tau) , \\
G^4(t,\tau) & := \frac{\varrho}{l(\tau,T)} \Psi^{14}(t,\tau) - \frac{K(\tau,T)}{2 l(\tau,T)}\Psi^{24}(t,\tau) - \Psi^{34}(t,\tau) . 
\end{split}
\end{align}
Next, we define the process
\begin{equation}
\Gamma^{\hat{v}}(t) := \frac{\Lambda'(t)}{\Lambda(t)} \left(P(t,T) + \tilde{M}(t) - 2 \phi \int_0^t X^{\hat{v}}(u) du \right),
\label{def:Gamma_quantity}
\end{equation}
where $\tilde{M}$ is the square integrable martingale 
\begin{equation}\label{def:square_int_mart_M_tilde}
\tilde{M}(s) := \mathbb{E}_s \bigg[2 \phi \int_0^\tau X^{\hat{v}}(u) du + 2 \varrho X^{\hat{v}}(\tau) - P(\tau,T) \bigg], 
\end{equation}
and $\mathbb{E}_t$ denotes the expectation conditioned on the filtration $\mathcal{F}_t$ for all $t \in [0,\tau]$. Finally we define the following functions on $0\leq t \leq \tau$,
\begin{align}\label{def:v_terms}
\begin{split}
v_0(t,\tau) & :=  \left(1-\frac{G^4(t,\tau) \Psi^{43}(t,\tau)}{G^3(t,\tau) \Psi^{44}(t,\tau)} \right)^{-1}, \\
v_1(t,\tau) & := \left(\frac{G^4(t,\tau) \Psi^{41}(t,\tau)}{G^3(t,\tau) \Psi^{44}(t,\tau)} - \frac{G^1(t,\tau)}{G^3(t,\tau)} \right), \\
v_2(t,\tau) & := \left(\frac{G^4(t,\tau) \Psi^{42}(t,\tau)}{G^3(t,\tau) \Psi^{44}(t,\tau)} - \frac{G^2(t,\tau)}{G^3(t,\tau)} \right) ,\\
v_3(t,\tau) & := \frac{G^4(t,\tau)}{G^3(t,\tau)}.
\end{split}
\end{align}
In order for the optimal strategy to be well defined, we will need additional assumptions. Note that if $l, K$ are positive constants these assumptions translate to Assumption 3.1 and Lemma 5.5 in \cite{neuman2020optimal}. 
\begin{assumption} \label{ass-opt}  
\begin{enumerate}[label=(A.\arabic*), ref=A.\arabic*] We assume that the following hold:
\item \label{A.1}
\[
\inf_{0\leq t \leq \tau} |G^4(t,\tau) \Psi^{43}(t,\tau) - G^3(t,\tau) \Psi^{44}(t,\tau)| > 0,
\] 
\item \label{A.2}
\[
\sup_{0 \leq t \leq \tau} |\Psi^{4j}(t,\tau)| < \infty, \ \ \sup_{0 \leq t \leq \tau} |G^j(t,\tau)| < \infty, \ \ \ j \in \left\{1,2,3,4\right\}
\]
\item \label{A.3}
\[
\inf_{0\leq t\leq \tau}  |\Psi^{44}(t,\tau)| >0, \quad  \inf_{0\leq t\leq \tau} | G^3(t,\tau)|>0.
\]
%\item \label{A.5} The time derivatives
%$\partial_t l(t,T)$ and $\partial_t K(t,T)$ exits and 
%$$
%\sup_{0\leq t \leq T}|\partial_t l(t,T)|<\infty, \ \ \ \sup_{0\leq t \leq T}|\partial_t K(t,T)|<\infty.
%$$
%Moreover, for any $0\leq \tau <T$, 
%\[
%0 < \sup_{0\leq t \leq \tau} l(t,T) < \infty, \ \ \ 0 < \sup_{0\leq t \leq \tau} K(t,T) < \infty
%\] 
\end{enumerate}
\end{assumption}
\begin{remark}
At this point we stress the fact that the conditions in Assumption \ref{ass-opt} are not very general, however the purpose of this section is to show how to incorporate optimal execution into the impacted bonds framework. Future work may improve the theoretical results on this topic. 
\end{remark} 
Next we present the main result of this section, which derives the unique optimal trading speed.

%------------------------ Theorem: Linear feedback form ------------------------%
\begin{theorem}[Optimal trading strategy]
\label{linear_feedback_form}
Under Assumption \ref{ass-opt}, there exists a unique optimal strategy $\hat{v} \in \mathcal{A}$ which maximises \eqref{opt_stoch_control} and it is given by the following feedback form 
\begin{align}\label{optimal_strategy_linear_feedback}
\begin{split}
v(t) & = v_0(t,\tau) \Bigg( v_1(t,\tau) X^v(t) + v_2(t,\tau) \Upsilon^v(t) \\
& \qquad \qquad + v_3(t,\tau) \mathbb{E}_t \left[ \int_t^\tau \frac{\Lambda(s) \Psi^{43}(s,\tau)}{\Psi^{44}(t,\tau)} (\mu(s) + \Gamma^{v}(s)) ds \right] \\
&  \qquad \qquad - \mathbb{E}_t \left[\int_t^\tau \Lambda(s) \frac{G^3(s,\tau)}{G^3(t,\tau)} (\mu(s) + \Gamma^{v}(s)) ds \right] \Bigg), 
\end{split}
\end{align}
for all $t \in (0,\tau)$.
\end{theorem}
The proof Theorem \ref{linear_feedback_form} is given in Section \ref{sec:proof_linear_feedback}.

%%%%%%%%%%%%%%%%%%%%%%%%%%%%%%%%%%%%%%%%%%%%%%%%%%%%%
%				Proofs main section 
%%%%%%%%%%%%%%%%%%%%%%%%%%%%%%%%%%%%%%%%%%%%%%%%%%%%%
\section{Proofs of the results from Section \ref{sec:main}}{\label{sec:proofs}}

%----------------------------------------------------------------------
%	Proof Theorem: Impacted market price of risk
%------------------- -----------------------------------------------
\begin{proof} [Proof of Theorem \ref{thm_impacted_market_price}]

We adapt the argument by Bjork in Section 3.2 of \cite{bjork1997interest} to our case. We fix two maturities $T$ and $S$, and we consider a portfolio $V$ consisting of $S$-bonds and $T$-bonds. 
We further assume that both bonds are traded with admissible trading speeds $v_T$ and $v_S$ which correspond by \eqref{impact_density} to impact densities $J_T$ and $J_S$.

From \eqref{class-b} and \eqref{impacted_bond_differential_form} we can write the dynamics of the impacted bonds as follows: 
\begin{equation}\label{dyn_imp_bonds}
\begin{aligned}
d \tilde{P}(t,T) = & \mu_T(t,r(t)) dt + J_T(t) dt + \sigma_T(t,r(t)) d W^\mathbb{P}(t), \\
d \tilde{P}(t,S) = & \mu_S(t,r(t)) dt + J_S(t) dt + \sigma_S(t,r(t)) d W^\mathbb{P}(t).
\end{aligned}
\end{equation}
Let $\tilde{h}_T,\tilde{h}_S$ by locally bounded predictable processes representing the weights of the $T$ and $S$ bonds, respectively. We denote by $\tilde{V}(t)$ the portfolio value process, i.e.
\begin{equation*}
\tilde{V}(t) \equiv \tilde{V}(t;\tilde{h}) := \tilde{h}_T(t) \tilde{P}(t,T) + \tilde{h}_S(t) \tilde{P}(t,S).
\end{equation*}
Since, by assumption, the impacted-portfolio is self-financing, it holds at any time $t$ (see Definition \ref{def:self_financing})
\begin{equation*}
d \tilde{V}(t;\tilde{h}) = \tilde{h}_T(t) d \tilde{P}(t,T) + \tilde{h}_S(t) d \tilde{P}(t,S).
\end{equation*}
It is convenient to define the relative (impacted) weights
\begin{equation*}
\alpha_{T_i}(t) := \frac{\tilde{h}_{T_i}(t) \tilde{P}(t,T_i)}{\tilde{V}(t;\tilde{h})},\ \ T_i \in \left\{T,S\right\},
\end{equation*}
We conclude that if the impacted portfolio is self financing, then
\begin{equation}
\frac{d \tilde{V}(t)}{\tilde{V}(t)} = \alpha_T(t) \frac{d \tilde{P}(t,T)}{\tilde{P}(t,T)} + \alpha_S(t) \frac{d \tilde{P}(t,S)}{\tilde{P}(t,S)}.
\label{self_fin_imp_portfolio_relative}
\end{equation}
In order to ease the notation, we suppress the dependence on $r(t)$ in the drift and volatility. Substituting the dynamics \eqref{dyn_imp_bonds} into \eqref{self_fin_imp_portfolio_relative}, we have
\begin{multline}
\frac{d \tilde{V}(t)}{\tilde{V}(t)} = \frac{\alpha_T(t)}{\tilde{P}(t,T)} (\mu_T(t) - J_T(t)) dt + \frac{\alpha_S(t)}{\tilde{P}(t,S)} (\mu_S(t) - J_S(t)) dt + \\
+ \left( \alpha_S(t) \frac{\sigma_S(t)}{\tilde{P}(t,S)} + \alpha_T(t) \frac{\sigma_T(t)}{\tilde{P}(t,T)} \right) d W^\mathbb{P}(t).
\label{relative_dynamics}
\end{multline}
At this point, we choose the relative weights so that the diffusive part of the equation above vanishes, that is, 
\begin{equation}\label{sys1}
\begin{aligned}
\alpha_T(t) + \alpha_S(t) & = 1, \\
\alpha_T(t) \frac{\sigma_T(t)}{\tilde{P}(t,T)} + \alpha_S(t) \frac{\sigma_S(t)}{\tilde{P}(t,S)} & = 0.
\end{aligned}
\end{equation}
Solving this system gives 
\begin{equation}\label{expression_relative_weights}
\begin{aligned}
\alpha_S(t) & = \frac{\sigma_T(t)/\tilde{P}(t,T)}{\sigma_T(t)/\tilde{P}(t,T) - \sigma_S(t)/\tilde{P}(t,S)}, \\
\alpha_T(t) & = - \frac{\sigma_S(t)/\tilde{P}(t,S)}{\sigma_T(t)/\tilde{P}(t,T) - \sigma_S(t)/\tilde{P}(t,S)}.
\end{aligned}
\end{equation}
Notice that the above expressions are well defined. Indeed, if the denominator was approaching zero, then the sum of the two weights would be zero and this would contradict \eqref{sys1}. Next, we substitute \eqref{expression_relative_weights} into \eqref{relative_dynamics}. Following again Bjork's argument, we use the fact that our impacted portfolio is locally risk-free (as in Definition \ref{def:locally_risk_free}) by assumption and deduce the following relationship must hold:
\begin{multline*}
\frac{\mu_T(t) - J_T(t)}{\tilde{P}(t,T)} \left(- \frac{\sigma_S(t)/\tilde{P}(t,S)}{\sigma_T(t)/\tilde{P}(t,T) - \sigma_S(t)/\tilde{P}(t,S)} \right) + \\
+ \frac{\mu_S(t) - J_S(t)}{\tilde{P}(t,S)} \left( \frac{\sigma_T(t)/\tilde{P}(t,T)}{\sigma_T(t)/\tilde{P}(t,T) - \sigma_S(t)/\tilde{P}(t,S)} \right) = r(t).
\end{multline*}
Multiplying both sides by the term
\begin{equation*}
\frac{\sigma_T(t)}{\tilde{P}(t,T)} - \frac{\sigma_S(t)}{\tilde{P}(t,S)},
\end{equation*}
we obtain
\begin{equation*}
\left(\frac{\mu_S(t) - J_S(t)}{\tilde{P}(t,S)} - r(t) \right) \left(\frac{\sigma_T(t)}{\tilde{P}(t,T)} \right) = \left(\frac{\mu_T(t) - J_T(t)}{\tilde{P}(t,T)} - r(t) \right) \left(\frac{\sigma_S(t)}{\tilde{P}(t,S)} \right).
\end{equation*}
It follows that,
\begin{equation*}
\left( \frac{\mu_S(t) - J_S(t)}{\tilde{P}(t,S)} - r(t) \right) \left( \frac{\tilde{P}(t,S)}{\sigma_S(t)}\right) = \left(\frac{\mu_T(t) - J_T(t)}{\tilde{P}(t,T)} - r(t) \right) \left( \frac{\tilde{P}(t,T)}{\sigma_T(t)}\right),
\end{equation*}
and rearranging we deduce
\begin{equation} \label{eq:lambda_tilde_independent_of_mat}
\frac{\mu_S(t) - J_S(t) - r(t) \tilde{P}(t,S)}{\sigma_S(t)} = \frac{\mu_T(t) - J_T(t) - r(t) \tilde{P}(t,T)}{\sigma_T(t)}.
\end{equation}
Notice that the left hand side of \eqref{eq:lambda_tilde_independent_of_mat} depends on $S$ but not on $T$, while the right hand side of \eqref{eq:lambda_tilde_independent_of_mat} depends on $T$, but not on $S$. Since $S$ and $T$ are arbitrary, we conclude that both sides of \eqref{eq:lambda_tilde_independent_of_mat} depend only on $t$ and $r(t)$.
\end{proof}

%----------------------------------------------------------------------
%			Proof Theorem: Absence of arbitrage
%----------------------------------------------------------------------
\begin{proof}[Proof of Theorem \ref{thm_absence_arbitrage}]
The proof is similar to the proof of Proposition 1.1 in Chapter 1.2 of \cite{bjork1997interest} (see also Harrison and Kreps \cite{harrison1979martingales} Theorem 2 and relative Corollary in Section 3 and Harrison and Pliska \cite{harrison1981martingales}, Theorem 2.7, Section 2). For the sake of completeness, we give the proof here, translated in our price impact environment. Let $T<+\infty$ be some finite maturity. Let $\tilde{h}$ be an arbitrage portfolio and $\tilde{V}$ the corresponding value process. Then, given the positivity of the discount factor (bank account) defined in \eqref{def:bank_account} and the equivalence between the real world measure $\mathbb{P}$ and the impacted risk neutral measure $\tilde{\mathbb{Q}}$, we immediately deduce
\begin{equation}
\tilde{\mathbb{Q}} \left( \frac{\tilde{V}(T)}{B(T)} \geq 0 \right) = 1, \ \ \ \tilde{\mathbb{Q}} \left( \frac{\tilde{V}(T)}{B(T)} >0 \right) > 0.
\label{Q_tilde_probs}
\end{equation}
Moreover we have
\begin{equation*}
0 = \tilde{V}(0) = \frac{\tilde{V}(0)}{B(0)} = \mathbb{E}^{\tilde{\mathbb{Q}}} \left[\frac{\tilde{V}(T)}{B(T)} \right] > 0,
\end{equation*}
where the first equality comes from the definition of arbitrage, the second from the fact that $B(0)=1$ and the third from the fact that $\tilde{V}(t)/B(t)$ is a martingale under $\tilde{\mathbb{Q}}$. Finally, the positivity of the expectation is a consequence of \eqref{Q_tilde_probs}. We get a contradiction so we conclude that absence of arbitrage must hold.
\end{proof}

%---------------------------------------------------------------------------
% Proof Proposition: Relationship between forward rate impact and zero-coupon bond impact
%---------------------------------------------------------------------------
\begin{proof}[Proof of Proposition \ref{prop:relationship_Jf_Jp}]
We start by writing the impacted forward rate defined in \eqref{impacted_forward_rate} as
\begin{equation*}
\tilde{f}(t,T) = f(t,T) + \int_0^t J^f(s,T) ds, 
\end{equation*}
where $f$ represents the unimpacted forward rate (see e.g. Chapter 6, of \cite{filipovic2009term}) and we used the assumption $\tilde{f}(0,t) = f(0,t)$. Then, using \eqref{impacted_bond_HJM}, we deduce
\begin{align}\label{eq:dummy}
\begin{split}
\tilde{P}(t,T) & = \exp \left\{-\int_t^T \tilde{f}(t,u) du\right\} \\
& = \exp \left\{-\int_t^T f(t,u) du - \int_t^T J^f(s,u) du \right\} \\
& = P(t,T) \exp \left\{-\int_t^T J^f(s,u) du \right\},
\end{split}
\end{align}
where $P$ denotes the unimpacted zero-coupon bond and we used the well known relation between $P(t,T)$ and $f(t,T)$. From \eqref{impacted_bond} and \eqref{def:overall_price_impact} we have 
\begin{equation} \label{bla} 
\tilde{P}(t,T) = P(t,T) - I_T(t).
\end{equation}
Substituting this last expression into \eqref{eq:dummy} and rearranging, we obtain
\begin{equation*}
\exp \left\{-\int_t^T J^f(s,u) du \right\} = \frac{\tilde{P}(t,T)}{\tilde{P}(t,T) + I_T(t)}.
\end{equation*}
By taking logarithms on both sides yields and using \eqref{bla} we get 
\begin{equation*}
\begin{aligned} 
\int_t^T J^f(s,u) du &= - \log \left( \frac{\tilde{P}(t,T)}{\tilde{P}(t,T) + I_T(t)} \right) \\
&= - \log \left( 1- \frac{I_T(t)}{P(t,T)} \right).
\end{aligned}  
\end{equation*}
Differentiating with respect to maturity, we get \eqref{relationship_Jf_Jp}.
\end{proof}

%----------------------------------------------------------------------
%			Proof Theorem: HJM condition with price impact
%----------------------------------------------------------------------
\begin{proof}[Proof of Theorem \ref{hjm_condition_market_impact} ]
Let $B(t)$ be the bank account defined in \eqref{def:bank_account} and let the impacted zero-coupon bond $\tilde{P}$ follow the dynamics \eqref{impaced_ZCB_HJM}. By applying Ito's formula to the discounted impacted zero-coupon bond price, we immediately find
\begin{equation*}
d \frac{\tilde{P}(t,T)}{B(t)} = \frac{\tilde{P}(t,T)}{B(t)} \tilde{b}(t,T) dt + \frac{\tilde{P}(t,T)}{B(t)} \nu(t,T) d W^{\mathbb{P}}(t),
\end{equation*}
with $\tilde{b}$ and $\nu$ defined as in \eqref{b-v-rel}. Changing measure form the real world $\mathbb{P}$ to the impacted risk neutral $\tilde{\mathbb{Q}}$ as in \eqref{def:Q_tilde_rnd_HJM} implies 
\begin{equation*}
d \frac{\tilde{P}(t,T)}{B(t)} = \frac{\tilde{P}(t,T)}{B(t)} \left(\tilde{b}(s,T) + \nu(t,T) \tilde{\gamma}(t) \right) dt + \frac{\tilde{P}(t,T)}{B(t)} \nu(t,T) d W^{\tilde{\mathbb{Q}}}(t).
\end{equation*}
Therefore, we clearly see that
\begin{equation*}
\frac{\tilde{P}(t,T)}{B(t)} \ \ \text{local martingale under}\ \ \tilde{\mathbb{Q}} \iff \tilde{b}(s,T) = - \nu(t,T) \tilde{\gamma}(t).
\end{equation*}
This is our new HJM condition. Notice also that since both functions $\nu$ and $\tilde{b}$ are continuous with respect to $T$, this condition is equivalent to saying that the impacted measure $\tilde{\mathbb{Q}}$ is an equivalent local martingale measure. Following Theorem  6.1 in \cite{filipovic2009term}, Chapter 6, we can plug in the explicit expression for $\tilde b$ in \eqref{b-v-rel} and write the HJM condition \eqref{HJM_condition} as
\begin{equation}\label{gh0} 
-\int_s^T \alpha(s,u) du - \int_s^T J^f(s,u) du + \frac{1}{2} \nu^2(s,T)  = - \nu(t,T) \tilde{\gamma}(t).
\end{equation} 
Differentiating both sides with respect to the maturity $T$ yields the equation 
\begin{equation*}
- \alpha(t,T) + \sigma(t,T) \int_t^T \sigma(t,u) du - J^f(t,T) = \sigma(t,T) \tilde{\gamma}(t),
\end{equation*}
that is
\begin{equation} \label{gh1} 
\alpha(t,T) + J^f(t,T) = \sigma(t,T) \int_t^T \sigma(t,u)du - \sigma(t,T) \tilde{\gamma}(t).
\end{equation}
Substituting \eqref{gh1} in the dynamics of the forward rate \eqref{impacted_forward_rate} and using Girsanov yields \eqref{impacted_forward_rate_Q_tilde}. Using \eqref{HJM_condition} along with \eqref{impaced_ZCB_HJM} and Girsanov gives \eqref{impacted_zc_bond_Q_tilde}.
\end{proof}

%%%%%%%%%%%%%%%%%%%%%%%%%%%%%%%%%%%%%%%%%%%%%%%%%%%%%
%           Proof theorem optimal execution
%%%%%%%%%%%%%%%%%%%%%%%%%%%%%%%%%%%%%%%%%%%%%%%%%%%%%
\section{Proof of Theorem \ref{linear_feedback_form}}{\label{sec:proof_linear_feedback}}
The uniqueness of the optimal strategy follows by a standard convexity argument for the performance functional \eqref{j-fun}. Hence we only need to derive the optimal strategy. 
  
We start by deriving a system of coupled forward-backward stochastic differential equations (FBSDEs) which is satisfied by the solution to the stochastic control problem. 

%------------------ Lemma: FBSDE system ---------------------------%
\begin{lemma}[FBSDE system]
\label{FBSDE_system}
A control $\hat{v} \in \mathcal{A}$ solves the optimization problem \eqref{opt_stoch_control} if and only if the processes $(X^{\hat{v}},\Upsilon^{\hat{v}},\hat{v},Z^{\hat{v}})$ satisfy the coupled forward-backward stochastic differential equations  
\begin{equation}
\begin{cases}
d X^{\hat{v}}(t) &= - \hat{v}(t) dt, \ \ \ X^{\hat{v}}(0) = x, \\
d \Upsilon^{\hat{v}}(t)& = - \rho \Upsilon^{\hat{v}}(t) dt + \gamma \hat{v}(t) dt, \ \ \ \Upsilon^{\hat{v}}(0) = y, \\
d \hat{v}(t) &= \Lambda(t) d P(t,T) - 2 \Lambda(t) \phi X^{\hat{v}}(t) dt \\
& \quad + \Upsilon^{\hat{v}}(t) \left[-\Lambda'(t) K(t,T) - \Lambda(t) \partial_t K(t,T) + \rho K(t,T) \Lambda(t) \right] dt \\
& \quad + Z^{\hat{v}}(t) \left[\Lambda'(t) + \rho \Lambda(t) \right] dt + \Lambda(t) \Gamma^{\hat{v}}(t) dt + d M(t), \\
 &\qquad \qquad \qquad \qquad \qquad\qquad\qquad \hat{v}(\tau) = \frac{\varrho}{l(\tau,T)} X^{\hat{v}}(\tau) - \frac{K(\tau,T)}{2 l(\tau,T)} \Upsilon^{\hat{v}}(\tau), \\
d Z^{\hat{v}}(t)& = \left(\rho Z^{\hat{v}}(t) + K(t,T) \gamma \hat{v}(t) \right) dt + d N(t), \ \ Z^{\hat{v}}(\tau) = 0,
\end{cases}
\end{equation}
for two suitable square integrable martingales $M=(M(\cdot,T))_{0 \leq t \leq \tau}$ and $N=(N(\cdot,T))_{0 \leq t \leq \tau}$, where the $\Lambda, \Gamma^{\hat{v}}$ and $\tilde{M}$ are defined in \eqref{def:Lambda_optimal_exe}, \eqref{def:Gamma_quantity} and \eqref{def:square_int_mart_M_tilde} respectively.
\begin{proof}
The proof follows the same lines as Lemmas 5.1 and 5.2 in \cite{neuman2020optimal}. Since for all $v \in \mathcal{A}$ the map $v \to \mathcal{J}(v)$ is strictly concave, we can study the unique critical point at which the Gateaux derivative of $\mathcal{J}$, which is defined as
\begin{equation*}
\langle \mathcal{J}'(v), \alpha \rangle := \lim_{\epsilon \to 0} \frac{\mathcal{J}(v+\epsilon \alpha) - \mathcal{J}(v)}{\epsilon},
\end{equation*}
is equal to $0$ for any $\alpha \in \mathcal{A}$. This derivative can be computed analytically as follows. Let $\epsilon>0$ and $v,\alpha \in \mathcal{A}$. Since for all $t \in [0,\tau]$,
\begin{align} \label{bb1}
\begin{split}
X^{v + \epsilon \alpha}(t) & = x - \int_0^t (v(s) + \epsilon \alpha(s) )ds = X^v(t) - \epsilon \int_0^t \alpha(s) ds \\
\Upsilon^{v+\epsilon \alpha}(t) & = \Upsilon^v(t) + \epsilon \gamma \int_0^t e^{-\rho (t-s)} \alpha(s) ds,
\end{split}
\end{align}
From \eqref{j-fun} and \eqref{bb1} we have
\begin{align*}
&\mathcal{J}(v+\epsilon \alpha) = \\
&= \mathbb{E} \Bigg[ \int_0^\tau \left( P(t,T) - K(t,T) \Upsilon^v(t) - K(t,T) \epsilon \gamma \int_0^t e^{-\rho (t-s)} \alpha(s) ds \right) \left(v(t) + \epsilon \alpha(t) \right) dt  \\
&\qquad - \int_0^\tau l(t,T) v^2(t) + \epsilon^2 l(t,T) \alpha_t^2 + 2 l(t,T) v(t) \epsilon \alpha(t) dt + X^v(\tau) P(\tau,T) - \epsilon P(\tau,T) \int_0^\tau \alpha(s) ds  \\
&\qquad - \phi \int_0^\tau (X^v(t))^2 + \epsilon \left(\int_0^t \alpha(s) ds \right)^2 - 2 X^v(t) \epsilon \int_0^t \alpha(s) ds dt  \\
&\qquad - \varrho \left( (X^v(\tau))^2 + \epsilon^2 \left(\int_0^\tau \alpha(s) ds \right)^2 - 2 X^v(\tau) \epsilon \int_0^\tau \alpha(s) ds \right) \Bigg]. 
\end{align*}
It follows that 
\begin{align*}
\mathcal{J}(v+\epsilon \alpha) - \mathcal{J}(v) &= \epsilon \mathbb{E} \Bigg[ \int_0^\tau \left( P(\tau,T) - K(t,T) \Upsilon^v(t) \right) \alpha(t) dt \\
&\qquad - \int_0^\tau K(t,T) v(t) \int_0^t \gamma e^{-\rho (t-s)} \alpha(s) ds dt - 2 \int_0^\tau l(t,T) v(t) \alpha(t) dt \\
&\qquad + 2 \phi \int_0^\tau X^v(t) \int_0^t \alpha(s) ds dt + 2 \varrho X^v(\tau) \int_0^\tau \alpha(s) ds
- P(\tau,T) \int_0^\tau \alpha(s) ds \Bigg]  \\
&\qquad+ \epsilon^2 \mathbb{E} \Bigg[ \gamma \int_0^\tau K(t,T) \alpha(t) \int_0^t e^{-\rho (t-s)} \alpha(s) ds dt - \int_0^\tau l^2(t,T) \alpha^2(t) dt  \\
&\qquad - \phi \int_0^\tau \left( \int_0^t \alpha(s) ds \right)^2 dt - \varrho \left(\int_0^\tau \alpha(s) ds \right)^2 \Bigg].
\end{align*}
Note that all the terms above are finite since $\ell$ and $K$ are bounded functions and since $\alpha, v \in \mathcal A$. Applying Fubini's theorem twice, we obtain
\begin{multline*}
\langle \mathcal{J}'(v), \alpha \rangle = \mathbb{E} \Bigg[ \int_0^\tau \alpha(s) \Bigg( P(s,T) - K(s,T) \Upsilon^v(s) - \int_s^\tau K(t,T) e^{-\rho(t-s)} \gamma v(t) dt + \\
- 2 l(s,T) v(s) + 2 \phi \int_s^\tau X^v(t) dt + 2 \varrho X^v(\tau) - P(\tau,T) \Bigg) ds \Bigg],
\end{multline*}
for any $\alpha \in \mathcal{A}$. We get the following condition on the optimal strategy
\begin{equation}
\begin{aligned}
\mathbb{E} \Bigg[ \int_0^\tau \alpha(s) \Bigg( P(s,T) - K(s,T) \Upsilon^v(s) - \int_s^\tau K(t,T) e^{-\rho(t-s)} \gamma v(t) dt  \\
- 2 l(s,T) v(s) + 2 \phi \int_s^\tau X^v(t) dt + 2 \varrho X^v(\tau) - P(\tau,T) \Bigg) ds \Bigg] = 0.
\label{first_order_condition}
\end{aligned}
\end{equation}
Next we show that given the optimal strategy $\hat{v} \in \mathcal{A}$, the vector $(X^{\hat{v}},\Upsilon^{\hat{v}})$ satisfies the first order condition \eqref{first_order_condition} if and only if the vector $(X^{\hat{v}},\Upsilon^{\hat{v}},\hat{v},Z^{\hat{v}})$ solves a FBSDE system, for some auxiliary process $Z$.

For any $s>0$ we denote by $\mathbb{E}_s$ the conditional expectation with respect to the filtration $\mathcal F_s$. Assume $\hat{v} \in \mathcal{A}$ maximizes the functional $\mathcal{J}$. Applying the optional projection theorem we obtain  
\begin{align*}
\mathbb{E} \Bigg[ \int_0^\tau \alpha(s) \bigg( P(s,T) - K(s,T) \Upsilon^v(s) - \mathbb{E}_s \bigg[\int_s^\tau K(t,T) e^{-\rho (t-s)} \gamma \hat{v}(t) dt \bigg] - 2 l(s,T) \hat{v}(s)  \\
+ \mathbb{E}_s \bigg[2 \phi \int_s^\tau X^{\hat{v}}(t) dt + 2 \varrho X^{\hat{v}}(\tau) - P(\tau,T) \bigg] \bigg) ds \Bigg] = 0,
\end{align*}
for all $\alpha \in \mathcal{A}$. This implies
\begin{equation}  \label{explicit_sol_BSDE_v_hat} 
\begin{aligned}
&P(s,T) - K(s,T) \Upsilon^{\hat{v}}(s) - e^{\rho s} \mathbb{E}_s \bigg[\int_s^\tau K(t,T) e^{-\rho t} \gamma \hat{v}(t) dt \bigg] - 2 l(s,T) \hat{v}(s) \\
&\qquad+ \mathbb{E}_s \bigg[2 \phi \int_s^\tau X^{\hat{v}}(t) dt + 2 \varrho X^{\hat{v}}(\tau) - P(\tau,T) \bigg] \\
&= P(s,T) - K(s,T) \Upsilon^{\hat{v}}(s) \\
&\qquad - e^{\rho s} \left( \mathbb{E}_s \bigg[ \int_0^\tau K(t,T) e^{-\rho t} \gamma \hat{v}(t) dt \bigg] - \int_0^s K(t,T) e^{-\rho t} \gamma \hat{v}(t) dt \right) \\
&\qquad - 2 l(s,T) \hat{v}(s) + \mathbb{E}_s \bigg[2 \phi \int_0^\tau X^{\hat{v}}(t) dt + 2 \varrho X^{\hat{v}}(\tau) - P(\tau,T) \bigg] - 2 \phi \int_0^s X^{\hat{v}}(t) dt \\
&= 0, \ \ \ d \mathbb{P} \otimes ds\ \ \text{a.e. on} \ \ \Omega \times [0,\tau]. \\
\end{aligned}
\end{equation} 
Next, we define the square-integrable martingale
\begin{equation}\label{def:square_int_mart_N_tilde}
\tilde{N}(s) := \mathbb{E}_s \left[\int_0^{\tau} K(t,T) e^{-\rho t} \gamma \hat{v}(t) dt \right]
\end{equation}
and the auxiliary square-integrable process
\begin{equation*}
Z^{\hat{v}}(s) := e^{\rho s} \bigg( \int_0^s K(t,T) e^{-\rho t} \gamma \hat{v}(t) dt - \tilde{N}(s) \bigg),
\end{equation*}
for all $s \in [0,\tau]$. Note that since both $l$ and $K$ are assumed to be uniformly bounded and $v \in \mathcal A$, we have that $P(\tau,T) \in L^2(\Omega,\mathcal{F}_\tau,\mathbb{P})$. Therefore, we obtain
\begin{equation} \label{fgr}
P(s,T) - K(s,T) \Upsilon^{\hat{v}}(s) + Z^{\hat{v}}(s) - 2 l(s,T) \hat{v}(s) + \tilde{M}(s) - 2 \phi \int_0^s X^{\hat{v}}(t) dt = 0
\end{equation}
almost everywhere on $\Omega \times [0,\tau]$, where $\tilde{M}$ is the square-integrable martingale defined in \eqref{def:square_int_mart_M_tilde}, and we immediately see that the process $Z^{\hat{v}}$ satisfies the BSDE
\begin{equation*}
d Z^{\hat{v}}(t) = \left(\rho Z^{\hat{v}}(t) + K(t,T) \gamma \hat{v}(t) \right) dt - e^{\rho t} d \tilde{N}(t), \quad Z^{\hat{v}}(\tau) = 0.
\end{equation*}
From \eqref{def:transient_impact} we get that $\Upsilon^{\hat{v}}$ satisfies   
\begin{equation*}
d \Upsilon^{\hat{v}}(t) = - \rho \Upsilon^{\hat{v}}(t) dt + \gamma \hat{v}(t) dt, \quad \Upsilon^{\hat{v}}(0) = y.
\end{equation*}
Recall that $\Lambda$ was defined in \eqref{def:Lambda_optimal_exe}. From \eqref{fgr} it follows that $\hat{v}$ satisfies the backward stochastic differential equation 
\begin{align*}
\begin{split}
d \hat{v}(s) & = \Lambda'(s) \left(P(s,T) - K(s,T) \Upsilon^{\hat{v}}(s) + Z^{\hat{v}}(s) + \tilde{M}(s) - 2 \phi \int_0^s X^{\hat{v}}(u) du \right)ds \\
& \quad + \Lambda(s) \bigg(d P(s,T) - \partial_s K(s,T) \Upsilon^{\hat{v}}(s) ds - K(s,T) d \Upsilon^{\hat{v}}(s) + d Z^{\hat{v}}(s) \\
& \quad \quad + d \tilde{M}(s) - 2 \phi X^{\hat{v}}(s) ds \bigg) \\
& = \Lambda'(s) \left(P(s,T) - K(s,T) \Upsilon^{\hat{v}}(s) + Z^{\hat{v}}(s) + \tilde{M}(s) - 2 \phi \int_0^s X^{\hat{v}}(u) du \right)ds \\
& \quad + \Lambda(s) d P(s,T) - \Lambda(s) \partial_s K(s,T) \Upsilon^{\hat{v}}(s) ds + \rho K(s,T) \Upsilon^{\hat{v}}(s) \Lambda(s) ds  \\
& \quad   + \Lambda(s) \rho Z^{\hat{v}}(s) ds  - 2 \Lambda(s) \phi X^{\hat{v}}(s) ds + \Lambda(s) d \tilde{M}(s) - \Lambda(s) e^{\rho s} d \tilde{N}(s)  \\
\hat{v}(\tau) & = \frac{\varrho}{l(\tau,T)} X^{\hat{v}}(\tau) - \frac{K(\tau,T)}{2 l(\tau,T)} \Upsilon^{\hat{v}}(\tau),
\end{split}
\end{align*}
Putting these equations together with \eqref{inv}, we obtain the FBSDE system \eqref{FBSDE_system} with $M,N$ square-integrable martingales defined as
\begin{align*}
\begin{split}
M(t) & := \int_0^t \Lambda(s) d \tilde{M}(s) - \int_0^t \Lambda(s) e^{\rho s} d \tilde{N}(s) \\
N(t) & := - \int_0^t e^{\rho s} d \tilde{N}(s).
\end{split}
\end{align*}
In order to check the integrability of $M$, recall that $\Lambda$ was defined in \eqref{def:Lambda_optimal_exe}. Since $l$ is bounded away from $0$ on $[0,\tau] $ (see \eqref{l-pos}) we have 
\[
\sup_{0 \leq t \leq \tau} | \Lambda(t) | < \infty.
\]
Then, it holds
\begin{align*}
\begin{split}
\mathbb{E}[M^2(t)] & \leq \mathbb{E} \left[ \int_0^t \Lambda^2(s) d [\tilde M]_s \right] + \mathbb{E} \left[ \int_0^t\Lambda^2(s) e^{2\rho s} d[\tilde N]_s \right] \\
& \leq C_1 \mathbb{E} [\tilde M]_T + C_2 \mathbb{E} [ \tilde N]_T \\
&< \infty
\end{split}
\end{align*}
for some constants $C_1,C_2$, where in the last inequality we used the fact that both $\tilde{M}$ and $\tilde{N}$ are square integrable martingales.

Next, assume that  $(\hat{v},X^{\hat{v}},\Upsilon^{\hat{v}},Z^{\hat v})$ is a solution to the FBSDE system \eqref{FBSDE_system} and $\hat{v} \in \mathcal{A}$. We will show that $\hat{v}$ satisfies the first order condition \eqref{first_order_condition}, hence it maximizes the cost functional \eqref{j-fun}. First, note that the BSDE for $\hat{v}$ can be solved explicitly and the solution is indeed given in \eqref{explicit_sol_BSDE_v_hat}
\begin{align*}
\begin{split}
\hat{v}(s) & = \frac{1}{2 l(s,T)} \Bigg(P(s,T)- K(s,T) \Upsilon^{\hat{v}}(s) + Z^{\hat{v}}(t) + \tilde{M}(s) - 2 \phi \int_0^s X^{\hat{v}}(t) dt \Bigg) \\
& = \frac{1}{2l(s,T)} \Bigg(P(s,T) - K(s,T) \Upsilon^{\hat{v}}(s) - e^{\rho s} \left(\tilde{N}(s) - \int_0^s K(t,T) e^{-\rho t} \gamma \hat{v}(t) dt \right) \\
& \quad + \tilde{M}(s) - 2 \phi \int_0^s X^{\hat{v}}(u) du \Bigg),
\end{split}
\end{align*}
with $\tilde{N},\tilde{M}$ defined in \eqref{def:square_int_mart_N_tilde} and \eqref{def:square_int_mart_M_tilde}, respectively. Plugging this into the first order condition \eqref{first_order_condition} yields
\begin{align*}
&\mathbb{E} \Bigg[\int_0^\tau \bigg(e^{\rho s} \left(\tilde{N}(s) - \int_0^\tau K(t,T) e^{-\rho t} \gamma \hat{v}(t) dt \right) - \tilde{M}(s) \\
&\qquad + 2 \phi \int_0^\tau X(t) dt + 2 \varrho X^{\hat{v}}(\tau) - P(\tau,T) \bigg)ds \Bigg] \\
&= \mathbb{E} \Bigg[ \int_0^\tau \alpha(s) \bigg( e^{\rho s} (\tilde{N}(s) - \tilde{N}(\tau)) - \tilde{M}(s) + \tilde{M}(\tau) \bigg) ds \Bigg] \\
&= \mathbb{E} \Bigg[ \int_0^\tau \alpha(s) \bigg(e^{\rho s} (\tilde{N}(s) - \mathbb{E}_s[\tilde{N}(\tau)]) - \tilde{M}(s) + \mathbb{E}_s[\tilde{M}(\tau)] \bigg)ds \Bigg] \\
& = 0,
\end{align*}
for all $\alpha \in \mathcal{A}$. Since $\tilde{N},\tilde{M}$ are martingales, hence the first order condition \eqref{first_order_condition} is satisfied and $\hat{v} \in \mathcal{A}$ is the optimal strategy.
\end{proof}
\end{lemma}
%------------------------------------------------------------------------------%

Before giving the proof of our main theorem, we will need the following Lemma, which will help us to show the optimal strategy in \eqref{linear_feedback_form} is indeed admissible.

%------------------- Lemma on the quantity Gamma -----------------------------%
\begin{lemma}\label{lem:Gamma_quantity}
Let $\Gamma^{\hat{v}}$ be defined as in \eqref{def:Gamma_quantity}. Then, there exist constants $C_1, C_2 > 0$ such that 
\begin{equation*}
\mathbb{E} \left[  \int_0^\tau \big(\Gamma^{\hat{v}}(s) \big)^{2}  ds \right] \leq C_1 + C_2\mathbb{E} \left[\int_0^\tau v^2(s) ds \right].
\end{equation*}
\begin{proof}
Firstly, by the assumptions on $l$ (see \eqref{l-pos} and \eqref{k-assump}) it follows that
\begin{equation}\label{bounded_ratio_lambda_lambda_prime}
\sup_{0 \leq t \leq \tau} \bigg | \frac{\Lambda'(t)}{\Lambda(t)} \bigg| = \sup_{0 \leq t \leq \tau} \bigg| \frac{\partial_t l(t,T)}{l(t,T)} \bigg| < \infty,
\end{equation}
where $\Lambda$ is given in \eqref{def:Lambda_optimal_exe}. Therefore, from \eqref{bounded_ratio_lambda_lambda_prime},  \eqref{def:Gamma_quantity} and Jensen's inequality we get that there exist constants $C_1,C_2>0$ such that
\begin{align*} 
 \mathbb{E} \left[  \int_0^\tau \Gamma^{\hat{v}}(s)^2  ds \right] 
&\leq \mathbb{E} \left[  \int_0^\tau \Big(\frac{\Lambda'(s)}{\Lambda(s)}\Big)^2 \left(P^2(s,T) + \tilde{M}^2(s) + 4 \phi^2 \Big(\int_0^s X^{\hat{v}}(u) du \Big)^2 \right) ds \right] \\ 
&\leq C_1 \mathbb{E} \left[ \int_0^\tau \left(  P^2(s,T) + \tilde{M}^2(s) +  4 \phi^2  \Big(\int_0^s X^{\hat{v}}(u) du \Big)^2 \right)ds \right] \\ 
&\leq C_2 +  4 \phi^2  \mathbb{E} \left[ \int_0^\tau \Big(\int_0^s X^{\hat{v}}(u) du \Big)^2 ds \right],
\end{align*}
where we used \eqref{def: mathcal_h_squared} and the fact that the martingale $\tilde{M}$ defined in \eqref{def:square_int_mart_M_tilde} is square-integrable. 

Next, using the definition of $X^{\hat{v}}$ in \eqref{def:X} and Jensen's inequality twice, we deduce
\begin{align*} 
\mathbb{E} \left[ \int_0^\tau \Big(\int_0^s X^{\hat{v}}(u) du \Big)^2 ds\right] & = 
\mathbb{E} \left[ \int_0^\tau \Big(\int_0^s \big(x-\int_0^u  \hat v(y) dy  \big)du\Big)^2 ds \right] \\
&\leq C_1 + C_2\mathbb{E} \left[\int_0^\tau \int_0^s \int_0^u \hat v^2(y) dy du ds \right] \\
&\leq C_1 + C_2 \mathbb{E} \left[ \int_0^\tau \int_0^\tau \int_0^\tau \hat v^2(y) dy du ds  \right] \\
&\leq C_1 + C_2 \mathbb{E} \left[\int_0^\tau  \hat v^2(y) ds \right], 
\end{align*} 
for some constants $C_1,C_2$, and we are done. 
\end{proof}
\end{lemma}
%------------------------------------------------------------------------------%
We are now ready to prove Theorem \ref{linear_feedback_form}. 

%-------------------- Proof Theorem: Linear feedback form ---------------------%
\begin{proof}[Proof of Theorem \ref{linear_feedback_form}] 
Define
\begin{equation*}
\mathbf{X}^{{v}}(t) := 
\left(
\begin{matrix}
X^{\hat v}(t) \\
\Upsilon^{\hat v}(t) \\
\hat v(t) \\
Z^{\hat v}(t)
\end{matrix}
\right), \ \ \ 
\mathbf{M}(t) :=
\left(
\begin{matrix}
0 \\
0 \\
P(t,T) + \int_0^t \Gamma^{\hat v}(s) ds + \int_0^t \Lambda^{-1}(s) d M(s) \\
\int_0^t \Lambda^{-1}(s) d N(s)
\end{matrix}
\right),
\end{equation*}
where $\Lambda$ and $\Gamma^{\hat v}$ are defined in \eqref{def:Lambda_optimal_exe} and \eqref{def:Gamma_quantity} respectively. The FBSDE system \eqref{FBSDE_system} can be written as
\begin{equation*}
d \mathbf{X}_t^{\hat v}= A(t) \mathbf{X}_t^{\hat v} dt + \Lambda(t) d \mathbf{M}(t), \ \ \ 0 \leq t \leq \tau,
\end{equation*}
where the matrix $A(t)$ is defined in \eqref{def:time_dep_matrix_A}, with initial conditions
\begin{equation*}
\mathbf{X}^{{\hat v},1}(0) = x, \ \ \ \mathbf{X}^{{\hat v},2}(0) = y,
\end{equation*}
and terminal conditions
\begin{equation} \label{term}
\left(\frac{\varrho}{l(\tau,T)}, - \frac{K(\tau,T)}{2 l(\tau,T)}, -1, 0 \right) \mathbf{X}^{\hat v}(\tau) = 0, \ \ \ (0,0,0,1) \mathbf{X}^{\hat v}(\tau) = 0.
\end{equation}
Exploiting linearity, the unique solution can be expressed as
\begin{equation*}
\mathbf{X}^{\hat v}(\tau) = \Phi(\tau) \Phi^{-1}(t) \mathbf{X}^{\hat v}(t) + \int_t^\tau \Phi(\tau) \Phi^{-1}(s) \Lambda(s) d \mathbf{M}(s),
\end{equation*}
where $\Phi$ solves the ODE \eqref{matrix_ODE}. Moreover, it can be immediately seen that the first terminal condition in \eqref{term} yields
\begin{align*}
\begin{split}
0 & = G^1(t,\tau) X^{\hat v}(t) + G^2(t,\tau) \Upsilon^{\hat v}(t) + G^3(t,\tau) \hat v(t) + G^4(t,\tau) Z^{\hat v}(t) \\
& \quad + \int_t^\tau \Lambda(s) \left(G^3(s,\tau) \left(d P(s,T) + \Gamma^{\hat v}(s) ds + \Lambda^{-1}(s) d M(s) \right) + G^4(s,\tau) \Lambda^{-1}(s) d N(s) \right)
\end{split}
\end{align*}
with $G = (G^1,G^2,G^3,G^4)$ defined in \eqref{def:G_components}. Solving for the trading speed $v$, taking expectations and using that $P \in \mathcal{H}^2$, together with the fact that both $M$ and $N$ are square integrable martingales, implies
\begin{align}\label{trading_speed_u_with_z}
\begin{split}
\hat v(t) & = - \frac{G^1(t,\tau)}{G^3(t,\tau)} X^{\hat v}(t) - \frac{G^2(t,\tau)}{G^3(t,\tau)} \Upsilon^{\hat v}(t) - \frac{G^4(t,\tau)}{G^3(t,\tau)} Z^{\hat v}(t) \\
& \quad - \mathbb{E}_t \left[\int_t^\tau \Lambda(s) \frac{G^3(s,\tau)}{G^3(t,\tau)}(\mu(s) + \Gamma^{\hat v}(s)) ds \right].
\end{split}
\end{align}
Recall that $\Psi$ was defined in \eqref{def:psi_matrix}. Then the second terminal condition in \eqref{term} implies
\begin{align*}
\begin{split}
0 & = (0,0,0,1) \Psi(t,\tau) \mathbf{X}^{\hat v}(t) + (0,0,0,1) \int_t^\tau \Psi(s,\tau) \Lambda(s) d \mathbf{M}(s) \\
& = \Psi^{41}(t,\tau) X^{\hat v}(t) + \Psi^{42}(t,\tau) \Upsilon^{\hat v}(t) + \Psi^{43}(t,\tau) \hat v(t) + \Psi^{44}(t,\tau) Z^{\hat v}(t) \\
& \quad + \int_t^\tau \Lambda(s) \left(\Psi^{43}(s,\tau) \left(d P(s,T) + \Gamma^{\hat v}(s) ds + \Lambda^{-1}(s) d M(s) \right) + \Psi^{44}(s,\tau) \Lambda^{-1}(s) d N(s) \right).
\end{split}
\end{align*}
Hence, taking expectation and solving for $Z^u$ yields
\begin{align}\label{aux_process_Z}
\begin{split}
Z^{\hat v}(t) & = - \frac{\Psi^{41}(t,\tau)}{\Psi^{44}(t,\tau)} X^{\hat v}(t) - \frac{\Psi^{42}(t,\tau)}{\Psi^{44}(t,\tau)} \Upsilon^{\hat v}(t) - \frac{\Psi^{43}(t,\tau)}{\Psi^{44}(t,\tau)} {\hat v}(t) \\
& \quad - \mathbb{E}_t \left[\int_t^\tau \frac{\Lambda(s) \Psi^{43}(s,\tau)}{\Psi^{44}(t,\tau)} (\mu(s) + \Gamma^{\hat v}(s)) ds \right].
\end{split}
\end{align}
Therefore, plugging \eqref{aux_process_Z} into \eqref{trading_speed_u_with_z} gives
\begin{align*}
\begin{split}
\hat v(t) & = - \frac{G^1(t,\tau)}{G^3(t,\tau)} X^{\hat v}(t) - \frac{G^2(t,\tau)}{G^3(t,\tau)} \Upsilon^{\hat v}(t) + \frac{G^4(t,\tau) \Psi^{41}(t,\tau)}{G^3(t,\tau) \Psi^{44}(t,\tau)} X^{\hat v}(t) \\
& \quad + \frac{G^4(t,\tau) \Psi^{42}(t,\tau)}{G^3(t,\tau) \Psi^{44}(t,\tau)} \Upsilon^{\hat v}(t) + \frac{G^4(t,\tau) \Psi^{43}(t,\tau)}{G^3(t,\tau) \Psi^{44}(t,\tau)} {\hat v}(t) \\
& \quad + \frac{G^4(t,\tau)}{G^3(t,\tau)} \mathbb{E}_t \left[\int_t^\tau \frac{\Lambda(s) \Psi^{43}(s,\tau)}{\Psi^{44}(t,\tau)} (\mu(s) + \Gamma^{\hat v}(s)) ds \right] \\
& \quad - \mathbb{E}_t \left[\int_t^\tau \Lambda(s) \frac{G^3(s,\tau)}{G^3(t,\tau)} (\mu(s) + \Gamma^{\hat v}(s)) ds \right]. 
\end{split}
\end{align*}
Rearranging and using the Definitions \ref{def:v_terms}, we obtain the linear feedback form \eqref{optimal_strategy_linear_feedback}. Finally, we prove that the optimal trading strategy is admissible, that is, $\hat{v} \in \mathcal{A}$, as defined in \eqref{def:admissset}. Thanks to assumptions  \eqref{A.1} and \eqref{A.2}, we immediately see that
\begin{equation*}
\sup_{0 \leq t \leq \tau} |v_0(t,\tau)| < \infty.
\end{equation*}
Similarly, from assumptions \eqref{A.1}-\eqref{A.3} we deduce that $v_1$ and $v_2$ are both bounded on $[0,\tau]$. Exploiting again assumptions \eqref{A.1}-\eqref{A.3}, together with \eqref{def: mathcal_h_squared} we get that  
\begin{align*}
& \sup_{0 \leq t \leq \tau} \bigg| \mathbb{E}_t \bigg[\int_t^\tau \frac{\Lambda(s) \Psi^{43}(s,\tau)}{\Psi^{44}(t,\tau)} (\mu(s) + \Gamma^{\hat v}(s)) ds 
 - \mathbb{E}_t \int_t^\tau \Lambda(s) \frac{G^3(s,\tau)}{G^3(t,\tau)} (\mu(s) + \Gamma^{\hat v}(s)) ds \bigg] \bigg|\\
& \leq C \mathbb{E} \left[\int_0^\tau( |\mu(s) |+ |\Gamma^{\hat v}(s)| )ds \right] \\
& \leq \tilde{C}_1 + \tilde{C}_2\left( \mathbb{E} \left[\int_0^\tau \Gamma^{\hat v}(s)^{2}ds \right] \right)^{1/2} \\
& \leq \tilde{C}_1 +\tilde C_2\left(\mathbb{E} \left[\int_0^\tau \hat{v}^2(s)ds \right] \right)^{1/2}, 
\end{align*}
where we have used Jensen's inequality and Lemma \ref{lem:Gamma_quantity} in the last two inequalities. 
Using the above bound, together with  equations \eqref{def:X} and \eqref{def:transient_impact} we get from  \eqref{optimal_strategy_linear_feedback} that
\begin{equation*}
\mathbb{E}[\hat{v}^2(t)] \leq C_1 + C_2 \int_0^\tau \mathbb{E}[\hat{v}^2(s)] ds, \ \ \ 0 \leq t \leq \tau,
\end{equation*}
for some positive constants $C_1,C_2$, where we used again Jensen's inequality. Thanks to Gronwall's lemma, we get that
\begin{equation*}
\sup_{0 \leq t \leq \tau} \mathbb{E}[\hat{v}^2(t)] < \infty,
\end{equation*}
which implies
\begin{equation*}
\int_0^\tau \mathbb{E}[\hat v^2(s)] ds < \infty.
\end{equation*}
Hence Fubini's theorem, we conclude that $\hat v\in \mathcal A$. 
\end{proof}
	
\section*{Acknowledgements}
We are very grateful to an anonymous referees for  careful reading of the manuscript,
and for a number of useful comments and suggestions that significantly improved this paper.

%%%%%%%%%%%%%%%%%%%%%%%%%%%%%%%%%%%%%%%%%%%%%%%%%%%%%%%%%%%%%%%%%%%%%%%
%								References
%%%%%%%%%%%%%%%%%%%%%%%%%%%%%%%%%%%%%%%%%%%%%%%%%%%%%%%%%%%%%%%%%%%%%%%
\newpage
 
%\bibliography{references}
%\bibliographystyle{siam}

\end{document}